%% file: 0-arxiv.tex
\documentclass[twoside]{article}

\usepackage{amsmath, amsfonts, amsthm}
\usepackage{nicefrac}
\usepackage{graphicx}
\usepackage{hyperref}
\usepackage{booktabs}
\usepackage{multirow}
\usepackage{siunitx}
\usepackage{natbib}
\usepackage[margin=1in]{geometry}

\usepackage{etoolbox}
\newtoggle{icml}
\togglefalse{icml}

\usepackage{comment}

\usepackage{wrapfig}
\usepackage{tabularray}
\usepackage[shortlabels]{enumitem}

\usepackage{thmtools}
\usepackage{thm-restate}

\newtheorem{theorem}{Theorem}
\numberwithin{theorem}{section}

\newtheorem{lemma}[theorem]{Lemma}
\newtheorem{claim}[theorem]{Claim}

\newtheorem{definition}[theorem]{Definition}
\newtheorem{assumption}[theorem]{Assumption}
\newtheorem{remark}[theorem]{Remark}

\usepackage{mathtools}

\input{z-macros}

\usepackage[dvipsnames]{xcolor}   

\usepackage{hyperref}
\hypersetup{
colorlinks=true,
linkcolor=DarkOrchid,
filecolor=magenta,
urlcolor=Orchid,
citecolor=CadetBlue
}

\usepackage[algo2e,ruled,noend]{algorithm2e}
\usepackage[font=small,tableposition=top]{caption}

\title{From Individual Experience to Collective Evidence: \\
A Reporting-Based Framework for Identifying Systemic Harms 
}
\author{Jessica Dai, Paula Gradu, Inioluwa Deborah Raji, Benjamin Recht\\{\textit{University of California, Berkeley}}
}
\date{February 2025}

\begin{document}

\maketitle

\thispagestyle{empty}
\begin{abstract}
When an individual reports  a negative interaction with some system, how can their personal experience be contextualized within broader patterns of system behavior?
We study the \textit{reporting database} problem, where individual reports of adverse events arrive sequentially, and are aggregated over time.
In this work, our goal is to identify whether there are subgroups---defined by any combination of relevant features---that are disproportionately likely to experience harmful interactions with the system.
We formalize this problem as a sequential hypothesis test, and identify conditions on reporting behavior that are sufficient for making inferences about disparities in true rates of harm across subgroups.
We show that algorithms for sequential hypothesis tests can be applied to this problem with a standard multiple testing correction.
We then demonstrate our method on real-world datasets, including mortgage decisions and vaccine side effects; on each, our method (re-)identifies subgroups known to experience disproportionate harm using only a fraction of the data that was initially used to discover them.
\end{abstract}

\input{1-intro}
\input{2-model}
\input{3-preponderance}
\input{4-algorithm}
\input{5-experiments}
\input{6-discussion}

\newpage
\section*{Acknowledgements}
We are grateful to Ian Waudby-Smith, Kevin Jamieson, and Robert Nowak for helpful discussions in developing this work. 

JD is supported in part by the National Science Foundation Graduate Research Fellowship Program under Grant No. DGE 2146752. Any opinions, findings, and conclusions or recommendations expressed in this material are those of the author(s) and do not necessarily reflect the views of the National Science Foundation. JD also thanks the AI Policy Hub at UC Berkeley for funding support in the 2023-2024 academic year. BR is generously supported in part by NSF CIF award 2326498, NSF IIS Award 2331881, and ONR Award N00014-24-1-2531. IDR is supported by the Mozilla Foundation and MacArthur Foundation. 

\bibliographystyle{plainnat}
\bibliography{z-biblio}

\newpage
\appendix
\input{a1-proofs}
\section{Further practical considerations}
\label{subsec:practical}
\input{a0-practical}
\end{document}

%% file: z-macros.tex
\newcommand{\E}{\mathbb{E}}

\newcommand{\X}{\mathcal{X}}
\newcommand{\1}{\mathbf{1}}

\newcommand{\Badevent}{Y}
\newcommand{\Logwealth}{\omega}
\newcommand{\Truereportrate}{{\gamma_G^{\text{TR}}}}
\newcommand{\Falsereportrate}{{\gamma_G^{\text{FR}}}}

\newcommand{\Basegroup}{{\mu_G^0}}
\newcommand{\Actualgroup}{{\mathrm{IR}_G}}

\newcommand{\Groups}{{\mathcal{G}}}
\newcommand{\FlagG}{{\mathcal {G}^{\text{Flag}}}}

\newcommand{\bigmid}{\, \bigg| \,}
\newcommand{\NullH}{{\mathcal{H}_0^G}}

\newcommand{\nicebonf}{{\nicefrac{\alpha}{|\Groups|}}}
\newcommand{\nicebonfinv}{{\nicefrac{|\Groups|}{\alpha}}}

\newcommand{\Gstar}{{G^\star}}

\newcommand{\Basestar}{{\mu_{\Gstar}^0}}

\newcommand{\RR}{\mathrm{RR}}

%% file: 1-intro.tex
\section{Introduction}

\iftoggle{icml}{}{
\begin{figure*}
\centering
    \includegraphics[width=0.9\linewidth]{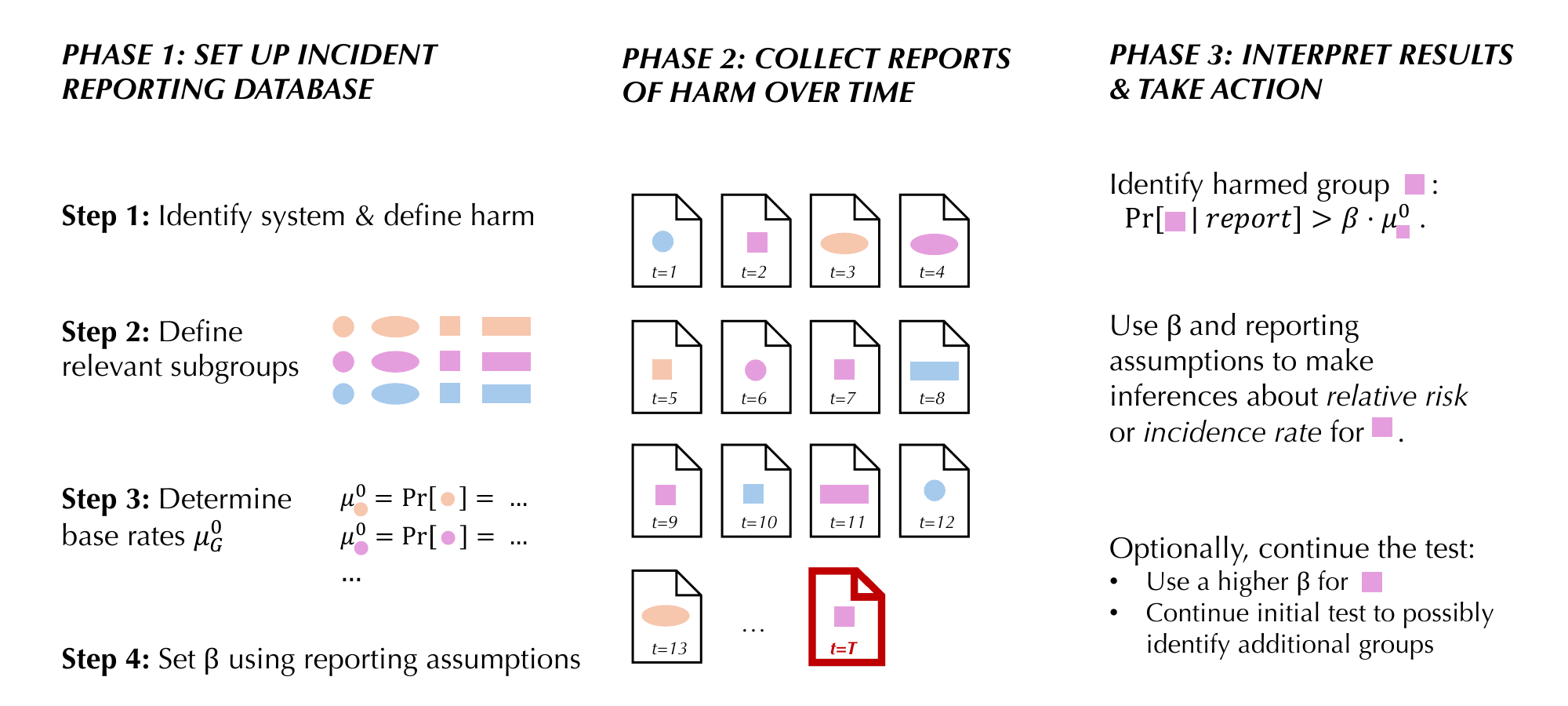}
    \caption{Overview of reporting database framework.}
    \label{fig:overview}
\end{figure*}
}
The impact of injustice is most acutely felt by the individual. But if an individual experiences harm, how can they know whether their experience is an isolated incident or part of a larger pattern of discrimination?

Fairness work has historically focused on model developers and third-party auditors as the main actors involved in creating fair mechanisms, motivating methods to construct models that are fair with respect to pre-defined subgroups at development time (e.g., as surveyed in \citet{pessach2022review})---or in identifying unfair ones, motivating post-hoc audits that occur after the entire decision-making process has completed (e.g., \citet{byun2024auditing, martinez2021secret}). 
However, in most applications where fairness is a concern, problems with the system may only emerge over time, and it is not necessarily obvious which subgroups might be important. Moreover, such approaches to fairness provide no mechanism for individuals to raise concerns. 
 
It is exactly this question of individual agency that drives our work. In addition to normative reasons, which suggest that individuals ought to have a voice in expressing concerns with their treatment (e.g., the literature on contestability of algorithmic decisions \citep{vaccaro2019contestability}), recent legislation has also highlighted individual reporting as a policy mandate for the governance of AI systems (e.g., the E.U. AI act \citep{EUAIAct2023}). While such legislation has yet to see full implementation,  mechanisms for individual incident reporting already exist in a variety of domains, including consumer finance, medical devices, and vaccines and pharmaceuticals. 
A key component of reporting databases in the latter settings is that information from individual reports are aggregated to build collective knowledge about specific vaccines or pharmaceuticals---and, when applicable, this aggregated information can drive downstream decisionmaking, such as updating vaccine guidelines or drug treatment protocols (e.g., \citet{oster2022myocarditis}). 

Fairness is an especially salient application for individual reporting systems: while individuals bear the harm, commonly-accepted (and legally-legible) notions of fairness are understood at an aggregate level. In fact, existing examples of (algorithmic) discrimination lawsuits  (e.g., 
\citet{Gilbert_2023} in hiring, or \citet{pazanowski}
in housing) are often structured as class actions, even as they are initiated by individuals based on their personal experiences. Crucially, individuals themselves may not know whether their experience with the system was inherently problematic, and deserving of redress, until it is placed in context with the experiences of others. On the other hand, while existing reporting databases do not typically analyze reporting behavior, it may be necessary to consider reporting more carefully in order for reporting databases to be useful for fairness auditing in more general settings, such as for algorithms that make allocation decisions.

In this paper, we consider what a realistic approach to assessing fairness claims from an incident reporting database might look like in practice. We are primarily interested in designing a framework for the general public to report and contest large-scale harms by leveraging reports of \emph{individual experience} to inform \emph{collective evidence} of discrimination. To this end, we propose \textit{reporting databases}, which allow individuals to submit reports of negative interactions, as a new mechanism for post-deployment fairness auditing. In particular, we identify conditions on reporting behavior and show how they can be used to to make inferences about rates of true harm in Section \ref{sec:interpretation}.
Our formalization of the problem allows us to leverage known approaches to sequential hypothesis testing. In Section \ref{sec:algs} we show how to instantiate two reasonable algorithms for our test and provide theoretical guarantees for each. Finally, in Section \ref{sec:experiments}, we illustrate the usefulness of our approach using real-world datasets, for applications with known disparity in per-subgroup rates of harm. On both real vaccine incident reports and on mortgage allocation decisions, our algorithm correctly identifies groups that disproportionately experience harm---and does so using a comparatively small number of reports.

\subsection{Related work \& application context}

The reporting database problem is at the intersection of various challenges addressed in fairness and statistics. 

\paragraph{Algorithmic accountability via (individual) reports.}
Some recent work considers methods for learning about fairness problems via individual reports from both theoretical \citep{globus2022algorithmic} and practical \citep{agostini2024bayesian} perspectives. However, most discussion of individual experiences in machine learning fairness literature is limited to contexts where the objective is to assess, appeal, contest or seek recourse for that individual to change their \textit{individual} outcomes, rather than forming a \textit{collective} judgment about the system as a whole~\citep{sharifi2019average, ustun2019actionable, karimi2022survey}.

Work on identifying fairness-related issues via reporting data has typically focused on learning in batch contexts, e.g. via 
positive-unlabeled learning for handling disparate reporting rates across subgroups (e.g., \citet{shanmugam2024quantifying, wu2022fairness}). 
In other works, identifying disparate reporting rates is itself is the central challenge (e.g., \citet{liu2022equity, liu2024quantifying}).
On the other hand, an emerging body of literature from the human-computer interaction community develops the concept of \textit{contestability} (e.g., \citet{almada2019human, vaccaro2019contestability, landau2024challenging}); though contestability is still typically understood in terms of individual outcomes, we see our work as one possible path to implementing this ideal\iftoggle{icml}{.}{, with an eye towards empowering contestability at larger scale.}

\paragraph{Fairness auditing as hypothesis testing.} 
\iftoggle{icml}{Existing proposals to formalize fairness auditing via hypothesis testing mainly consider batch settings (i.e. post-hoc or pre-deployment) \citep{cenai, cherian2023statistical}.}{\citet{cenai} make a direct connection between legal AI fairness audit requirements and hypothesis testing, though they mainly consider a post-hoc setting. \citet{cherian2023statistical} take a multiple testing approach for handling a large number of groups, but this test is again post-hoc (or entirely pre-deployment).} Two more closely related works are that of \citet{chugg2024auditing} and \citet{feng2024monitoring}, who propose applying sequential hypothesis tests with the explicit goal of identifying problems in deployed systems in real time.
However, as neither of these works study a reporting model, we propose fundamentally different tests: they test equality of means across different groups, while we compare within groups. 

\iftoggle{icml}{}{
\input{a0-additional-related}
}

%% file: a0-additional-related.tex
\paragraph{Identifying and defining subgroups.}
One approach to subgroup definition, following the line of work in multicalibration \cite{hebert2018multicalibration}, is to simply enumerate over all possible combinations of covariates. 
For sequential problems, per-group guarantees can be provided for subgroups that are learned online \citep{dai2024learning}, though these guarantees are in terms of prediction quality rather than statistical validity. For sequential experiments, \citet{adam2024should}  propose an approach to early stopping that does not require the experimenter to pre-specify the group experiencing harm, but instead identifies those who appear to be harmed more frequently. Though this is in spirit similar to the idea of identifying groups that report more frequently, their algorithm (and application context) is substantially different.

\paragraph{Sequential and multiple testing with anytime guarantees.}
One of our proposed tests provides anytime-validity guarantees by adapting the analysis of \cite{jamieson2014lil} and \cite{balsubramani2014sharp}. 
Our second proposed test leverages the recent literature on e-values (e.g. \cite{waudby2024estimating,vovk2021values}), which can be used to construct sequential tests that have validity guarantees in finite samples. While existing literature suggests methods for global null testing that can aggregate e-processes (e.g., \citet{choPeekingPEAKSequential2024} or \citet{chi2022multiple}), such approaches are unable to provide per-hypothesis guarantees.

\paragraph{Application \& policy context.} 
Sequential hypothesis tests have been used for real-world monitoring of adverse incidents in vaccines and medical devices (see, e.g., \cite{shimabukuro2015safety}).  Descriptive studies have identified disparate adverse impacts in pharmaceutical ~\citep{lee2023gender,whitley2009sex} and vaccine settings~\citep{oster2022myocarditis}. 
In AI policy, there have already been several calls to adopt a post-market surveillance regime for AI governance (e.g., \citet{raji2022outsider}).
% \jesscomment{Some of these are conflating the individual incident reporting model with the accident catalog model....}
% On March 21 2024, backed by 120 member states, 
The U.N. General Assembly's first AI Resolution (7 8/265 and 78/311) explicitly encourages ``the incorporation of feedback mechanisms to allow evidence-based discovery and reporting by end-users and third parties of 
% technical vulnerabilities and, as appropriate, 
[...]
misuses of artificial intelligence systems and artificial intelligence incidents'' 
% following their development, testing and deployment'' 
\citep{un-ai-res}. %https://documents.un.org/doc/undoc/ltd/n24/065/92/pdf/n2406592.pdf
In the U.S., Biden's (now repealed) AI Executive order explicitly directs the Department of Health and Human Services (HHS) to ``establish a [...]
% common framework for approaches to identifying and capturing clinical errors resulting from AI deployed in healthcare settings as well as specifications for a 
central tracking repository for associated incidents that cause harm, including through bias or discrimination''
% , to patients, caregivers, or other parties''
~\citep{biden2023executive}. In the E.U., Chapter IX of the 2024 EU AI Act focuses on post-market surveillance, with Articles 85 and 87 specifically highlighting individual reporting of harms.

\paragraph{Key definitions \& clarifications.}
Finally, we note that for AI systems, the term ``incident database'' has been used to describe systems for monitoring the adverse impact of algorithmic deployments, 
which often take the form of accident catalogs that focus on one-off, large-scale events (e.g.,~\citet{feffer2023ai,raji2022outsider,mcgregor2021preventing, ojewale2024towards, turri2023we}). 
However, in the context of our work, we are actively excluding these accident catalog databases. Instead, we focus on reporting databases that provide records of individual experiences of adverse events that are tied to specific systems. 

%% file: 2-model.tex
\section{Model, Notation, and Preliminaries}
\label{sec:model}
The goal of constructing a reporting database is to determine whether some system that individuals interact with---for example, an (algorithmic) loan decision system, or a medical treatment---results in disproportionate harm to some meaningful subgroups. 
For the reporting database associated with a particular system, we will use $\Badevent \in \{0,1\}$ as an indicator variable that denotes the undesirable event corresponding to that system. For example, in loan decisions, this could correspond to the event that a highly-qualified individual was denied a loan; in the medical setting, this may be an adverse physical side effect due to the treatment. 

\paragraph{Subgroup definitions.}
Individuals are characterized with feature vectors $X \in \X$, and we index individuals as 
$X_i$ (``features of individual $i$'') or 
$X_t$ (``features of the individual who reports at time $t$'').
Every individual $X_i$ ``belongs to'' at least one group $G$, and we will denote the event that $X_i$ belongs to $G$ as $\{X_i \in G\}$; we will use $\Groups$ to denote the set of all possible groups. 
This set of possible groups $\Groups$ can be defined arbitrarily as long as all groups can be determined as a function of covariates $\X$. We allow for groups to be overlapping---that is, we allow each individual $X_i$ to be in multiple groups so that $|\{G' \in \Groups: X_i \in G'\}| \geq 1$. 
\iftoggle{icml}{}{For example, it is possible to set $\Groups := 2^{\X}$ as in \citet{hebert2018multicalibration}. }

\paragraph{Reference population.}
The system for which the database is constructed naturally has a corresponding reference population of eligible individuals. For example, this could be everyone who has applied for a loan, or everyone who has been prescribed a certain medication. Thus, given a set of groups $\mathcal G$, we assume that it is possible to compute the composition of the reference population. 
\begin{assumption}[Reference population]
\label{assn:ref}
   For every $G \in \Groups$, the quantity $\Basegroup := \Pr[X \in G]$ is known. Throughout this work, we refer to the set $\{\Basegroup\}_{G \in \Groups}$ as \emph{base preponderances}.
\end{assumption}

\paragraph{Probabilistic model of reporting.}
As the database administrator, the high-level goal is to determine whether there exists some subgroup $G \in \Groups$ where $\Pr[\Badevent\mid X \in G]$ is abnormally high. 
Crucially, the database does not have access to information about every individual who interacts with the system; instead, individuals \textit{may} report to the database if they believe that they experienced bad event $\Badevent$. 
We thus let $R_i$ be a random variable representing whether individual $i$ decides to report (with $R_i = 0$ indicating no report). 

Each report $X_t$ is received sequentially, and assumed to be sampled i.i.d. from some underlying reporting distribution.\iftoggle{icml}{}{\footnote{Note that the events $\{X_t \in G\}$ and $\{X_t \in G'\}$ are correlated for any $G, G' \in \Groups$, i.e. the independence does not hold across groups. The key point in our case is independence across time.}}
Given a group $G$, we denote its corresponding mean among reports $\Pr[X_t \in G \mid R_t = 1]$ as  $\mu_G$.
We will sometimes refer to $\{\mu_G\}_{G \in \Groups}$ as (reporting) preponderances, as they represent the proportion of \textit{reports} that each $G$ comprises. A central claim of this paper is that comparing $\mu_G$ to $\Basegroup$---i.e., the extent to which group $G$ is (over)represented within the reporting database---can be a useful signal for $\Pr[Y \mid G]$ in a wide class of applications.\footnote
{Because we allow groups to overlap, we cannot enforce $\sum_G \Basegroup = 1$ or $\sum_G\mu_G = 1$.} The i.i.d. model of course simplifies the analysis and exposition, but itself is not intrinsic to modeling the reporting database problem as a sequential hypothesis test. As we discuss in Appendix \ref{subsec:practical}, the explicit i.i.d. assumption can be relaxed; more generally, any probabilistic model for sequential testing can be adapted to reporting databases. 

%% file: 3-preponderance.tex
\section{Identifying Discrimination by Modeling Preponderance}
\label{sec:interpretation} 
A major challenge of assessing potentially-differential rates of harm across subgroups using only reporting data is to relate the event that someone submits a report to the event that they experienced harm. That is, if someone did experience a negative outcome, how likely is it for them to have reported it, and conversely, if someone submitted a report, how likely is it to reflect ``true'' harm? Moreover, as is known from prior work, reporting rates themselves can vary across subgroups. 

Our central proposal is to conduct a hypothesis test for each group to determine whether it is overrepresented by a factor of $\beta$ among reports. That is, for each $G \in \Groups$, we test the following hypotheses:
\begin{align}
\label{eq:htest}
    \NullH: \mu_G < \beta\Basegroup && \mathcal{H}_1^G: \mu_G > \beta\Basegroup.
\end{align}
In Section \ref{sec:algs}, we will discuss concrete algorithms for conducting this test sequentially and their corresponding theoretical guarantees. 
Before doing so, we first argue that testing for preponderance among reports, i.e., tracking $\mu_G$ in this way, can be a meaningful way to identify discrimination, even when exact reporting behavior is unknown. 
In Sections \ref{subsec:rr} and \ref{subsec:ir}, we describe two distinct ways that this particular test can be interpreted; in Appendix \ref{subsec:practical}, we discuss some practical considerations for the modeling task.

\subsection{Preponderance as relative risk} 
\label{subsec:rr}
The first interpretation of our test allows us to make inferences about relative risk, the ratio between the rate of harm experienced by group $G$ and on average over the population. 
In this interpretation, the key quantity is the \textit{report-to-incidence ratio}.
\begin{definition}[Report-to-incidence ratio]
\label{def:rir}
We define the \emph{report-to-incidence-ratio (RIR)} as $\rho := \tfrac{\Pr[R = 1]}{\Pr[Y = 1]}$, and the group-conditional analogue as $\rho_G := \tfrac{\Pr[R = 1 \mid G ]}{\Pr[Y = 1 \mid G]}.$
\end{definition}

In Proposition \ref{prop:relativerisk-conversion}, we show that if the group-conditional RIR of some group $G$ is at most some constant multiple of the population-wide RIR, then we can convert a lower bound on report preponderance into a lower bound on true relative risk\iftoggle{icml}{ (see Appendix \ref{app:seq-proofs} for proof)}{}.
\begin{restatable}{proposition}{propRR}
\label{prop:relativerisk-conversion}
Define the relative risk of group $G$ to be $\RR_G := \frac{\Pr[Y = 1 \mid G]}{\Pr[Y = 1]}$. 
Suppose that for some group $G$ we have $\rho_G \leq b \cdot \rho$. Suppose that we determine that $\mu_G \geq \beta\Basegroup$ for some $\beta > 1$. Then, the true relative risk experienced by $G$ is at least $\RR_G \geq \nicefrac{\beta}{b}$.
\end{restatable} 
\iftoggle{icml}{}{
\begin{proof}
First, note that by definition of $\rho$, $\rho_G$, and $\RR_G$, we have 
\[
\rho_G \leq b \cdot \rho \iff \frac{\Pr[R = 1 \mid G ]}{\Pr[Y = 1 \mid G]} \leq b \cdot \frac{\Pr[R = 1]}{\Pr[Y = 1]} \iff \RR_G \geq \frac{\Pr[R = 1 \mid G]}{\Pr[R = 1]} \cdot \frac{1}{b}. 
\]
By Bayes' rule, $\frac{\Pr[R = 1 \mid G]}{\Pr[R = 1]} = \frac{\Pr[ G \mid R = 1]}{\Pr[G]} = \frac{\mu_G}{\Basegroup}$; furthermore, by assumption, we have $\frac{\mu_G}{\Basegroup} \geq \beta$. 
The result follows from combining with the previous display. 
\end{proof}}
Suppose we take $\max_G \nicefrac{\rho_G}{\rho} \leq b = 1.25$, i.e., no group over-reports 25\% more often than the population average. Then, if a test identifies a group $G$ for which $\mu_G \geq 1.75 \cdot \Basegroup$, this implies that the true relative risk for group $G$ is at least $\RR_G \geq 1.4$---that is, $G$ experiences harm 40\% more frequently relative to the population average. 

\subsection{Preponderance as true incidence rate}
\label{subsec:ir}
We now discuss an alternate way to convert a lower bound on preponderance into a guarantee on real-world harm. In this case, we can infer the true incidence rate of harm (that is, no longer relative to the average) if we are able to estimate---or willing to make assumptions on---true and false reporting behavior in groups.
Moreover, assumptions (or estimations) of these reporting rates need only be made in relation to the population average reporting rate $\Pr[R]$. 

\begin{definition}[Reporting rates]
Let $r := \Pr[R]$ be the average reporting rate over the full population.
    Let $\Truereportrate := \frac
    1r\Pr[R_i = 1 \mid \Badevent_i =1, X_i \in G]$, 
    $\Falsereportrate := \frac1r\Pr[R_i = 1 \mid \Badevent_i = 0, X_i \in G]$.
    Finally, let
    $\Actualgroup := \Pr[\Badevent \mid G]$ represent the true incidence rate, i.e. the likelihood that an individual in $G$ experiences $Y$.
\end{definition}
Note that $r \cdot \Truereportrate$ represents the (possibly group-conditional) rate at which individuals $X_i \in G$ who experience $\Badevent$ actually report, while $r \cdot \Falsereportrate$ represents the rate that individuals $X_i \in G$ who do not experience $\Badevent$ report.
Thus, $\Truereportrate$ and $\Falsereportrate$ represent how much more (or less) a particular group $G$ makes true or false reports relative to how much the whole population reports on average (which includes both true and false reports). 
The following proposition makes the relationships between $\Truereportrate$, $\Falsereportrate$, and our quantity of interest $\Actualgroup$, more precise.
\begin{restatable}{proposition}{propIR}\label{prop:reporting-conversion} 
Suppose that, for some $G$, it is determined that $\mu_G \geq \beta\Basegroup$ for some $\beta > 1$. As long as $\Truereportrate > \Falsereportrate$ for every $G \in \Groups$, 
% \[
$\Actualgroup \geq \frac{\beta - \Falsereportrate}{\Truereportrate - \Falsereportrate}.$
% \]
\end{restatable}
\iftoggle{icml}{See Appendix \ref{app:seq-proofs} for the (short) proof.}{
\begin{proof}[Proof of Proposition \ref{prop:reporting-conversion}]
Recall that we have defined $\mu_G = \Pr[G \mid R]$, and $\Basegroup = \Pr[G]$ is known by Assumption \ref{assn:ref}.
By Bayes' rule, we have
$    \mu_G = \Pr[G \mid R] = \frac{\Pr[G]\Pr[R \mid G]}{\Pr[R]} =  \Basegroup\frac{\Pr[R \mid G]}{r},$
Now, let us decompose $\Pr[R \mid G]$ by ``true'' reports ($\Badevent = 1$) and ``false'' reports ($\Badevent = 0$). 
By the law of total probability,
$
    \Pr[R \mid G] 
    = r \cdot \left(\Truereportrate \Actualgroup + \Falsereportrate(1-\Actualgroup )\right)
$; more precisely, 
\begin{align*}
    \frac1r\Pr[R \mid G] &= \Pr[R \mid G, \Badevent = 1]\Pr[\Badevent \mid G]  + \Pr[R \mid G,  \Badevent =0](1-\Pr[\Badevent \mid G] )
    \\&= \Truereportrate \Actualgroup + \Falsereportrate(1-\Actualgroup )
    \\&= \Falsereportrate + \Actualgroup (\Truereportrate-\Falsereportrate);
\end{align*} 
combining this with the Bayes' rule computation, cancelling the $\frac1r$ factor, gives us $
    \Actualgroup  = \frac{\frac{\mu_G}{\Basegroup} - \Falsereportrate}{\Truereportrate-\Falsereportrate}.
$
The result follows from the assumption that $\nicefrac{\mu_G}{\Basegroup} \geq \beta.$
\end{proof}
}
Proposition \ref{prop:reporting-conversion} shows that the exact computation of $\Actualgroup$ depends on reporting rates $\Truereportrate$ and $\Falsereportrate$. While these quantities are not directly estimable from reporting data---in fact, estimating reporting rates is itself a distinct research challenge (see, e.g., \citet{liu2024quantifying})---these results can nevertheless guide qualitative interpretation of how severe $\Actualgroup$ is. 

For example, suppose a test is run for $\beta = 1.5$. 
Suppose $G$ overreports relative to the population average, with $\Falsereportrate = 1$, and $\Truereportrate = 2$. 
\iftoggle{icml}{}{That is, $G$ \textit{falsely reports} at the same rate as the population reports on average (which includes both true and false reports), and submits true reports at twice the population average rate.} Under these (generous) assumptions, we will have $\Actualgroup = 0.5$, an extremely high incidence rate for any application---regardless of incidence rates for other groups. 

Alternatively, suppose reporting rates did not vary by group (i.e., $\Truereportrate = \gamma^{\text{TR}}$ and $\Falsereportrate = \gamma^{\text{FR}}$ for all $G$). Then, we can lower bound the disparities between true incidence rates across groups: 
if $G$ is flagged at $\beta > 1$, there must be some other group $G'$ with $\Actualgroup - \mathrm{IR}_{G'} \geq \frac{\beta - 1}{\gamma^{\text{TR}}-\gamma^{\text{FR}}}$. If it is further assumed that $\gamma^{\text{FR}}=0$, then $\Actualgroup - \mathrm{IR}_{G'} \geq \beta - 1$.

%% file: 4-algorithm.tex
\section{Identifying Subgroups with High Reporting Overrepresentation}
\label{sec:algs}

How might the test proposed in Equation \eqref{eq:htest} be carried out in practice, with reports arriving over time, and what properties might we want for such a test? In this section, we provide two ways to instantiate this sequential hypothesis test. For each, we provide two types of guarantees. The first is (sequential) $\alpha$-validity, which, roughly speaking, guarantees correctness of groups identified in $\FlagG$. More formally, we say that a sequential test is valid for a single group $G$ at level $\alpha$ if $\Pr[\exists t: \NullH \; \mathrm{ rejected}] \leq \alpha$ when $\NullH$ is true. Because we are testing for all groups in $\Groups$ simultaneously, we say that a sequential test is valid with respect to all groups $\Groups$ if $\Pr[\exists t, \exists G: \NullH \; \mathrm{erroneously} \; \mathrm{rejected}] \leq \alpha$.

The second type of guarantee is power, which guarantees that the test will identify a harmed group, if one exists. In particular, we are interested in the \textit{stopping time} $T$ of the test, which is the number of samples required for the test to reject the first null, i.e. to raise an alarm for any group.

At a high level, our algorithms for 
conducting this test
follow the protocol outlined in Algorithm \ref{alg:abstract}. 
\iftoggle{icml}{}{Every report $X_t$ can be considered a binary vector indexed by the groups in $\Groups$; the $G$ component of this vector is equal to $\1[X_t \in G]$. If there was only one group, we could run a sequential hypothesis test to determine whether $\mu_G$ was unacceptably large. With multiple groups, we can run $|\Groups|$ separate sequential hypothesis tests in parallel, one for each group, and correct the confidence levels for multiple hypothesis testing.}
For each group $G$, we maintain a test statistic $\Logwealth_t^G$ that is updated as reports $X_t$ are received over time. At each time $t$, each of these test statistics are compared to a threshold $\theta_t(\alpha)$, which depends on the test level $\alpha$; the null hypothesis $\NullH$ for group $G$ is rejected if $\Logwealth_t^G > \theta_t(\alpha)$. 
For ease of exposition, Algorithm \ref{alg:abstract} is written so that groups corresponding to rejected nulls are collected in a set $\FlagG$; in practice, a database administrator may choose to stop the test entirely as soon as one harmed group has been found.  

Correcting for multiple hypothesis testing across groups is handled by a simple Bonferroni correction---that is, given a particular test level $\alpha$, we test each individual group $G$ at level $\nicebonf$ rather than level $\alpha$.
% \footnote{The Bonferroni correction is also convenient for dealing with the correlations across hypotheses.} 
Though Bonferroni corrections often seem onerous in non-sequential settings, we show that, for sequential problems, the Bonferroni correction incurs only a modest increase in stopping time. 

In Section \ref{subsec:ztest}, we give a simple sequential Z-test-inspired approach which leverages a finite-time Law of the Iterated Logarithm. Section \ref{subsec:e-val-alg} presents a more complicated algorithm that leverages recent developments in anytime-valid inference.
The main differences in each algorithm lie in how they implement Lines 1 and 6 of Algorithm \ref{alg:abstract}---that is, how test statistics and thresholds are computed. 
For each instantiation of Algorithm \ref{alg:abstract}, we show validity and power guarantees.
Omitted proofs are given in Appendix \ref{app:seq-proofs}.

\begin{algorithm2e}
\caption{General protocol for testing overrepresentation}\label{alg:abstract}
\LinesNumbered
\KwIn{Set of groups $\Groups$; base preponderances $\{\Basegroup\}_{G \in \Groups}$; test level $\alpha$; relative strength $\beta$}
Initialize test statistic $\Logwealth_0^G$ for every $G \in \Groups$ and set threshold $\theta_0(\alpha)$\;
Initialize set of rejected nulls (flagged groups) $\FlagG := \emptyset$\;
\For{$t = 1, 2, \dots$}
{
See report $X_t$\;
\For{$G \in \Groups$}
{   Update test statistic $\Logwealth_t^G$ and compute threshold $\theta_t(\alpha)$\;
    \If{$\Logwealth_t^{G} \geq \theta_t(\alpha)$}{
    Add $G$ to $\FlagG$ and take requisite action for $G$, if applicable.}
    }
}
\end{algorithm2e}

\subsection{Sequential Z-test}
\label{subsec:ztest}
One simple observation that arises from the model presented in Section \ref{sec:model} is that if each report $X_t$ is drawn i.i.d. from some underlying distribution, then one might expect to be able to use concentration as a tool to conduct this test, since as time passes, the fraction of reports within the database from group $G$ should converge to the true mean $\mu_G$. We refer to this style of approach as a sequential Z-test, as it relies on measuring deviation from the mean. 

\paragraph{Updating the test statistic $\Logwealth_t^G$.}
Given this intuition, the test statistic is a simple count of the number of times a report from each group has been seen, i.e. (with $\Logwealth_0^G = 0$),
\begin{equation}
    \label{eq:ztest_update}
    \Logwealth_t^G \leftarrow \Logwealth_{t-1}^G +\1[X_t\in G]. 
\end{equation}
\paragraph{Setting the threshold $\theta_t(\alpha)$.}
Given the way that $\Logwealth_t^G$ accumulates evidence, one natural way to construct the threshold at each $t$ is to use the mean under the alternative, plus a correction term for both sample complexity and repeated testing over time. With $C$ set to either $\sqrt{\beta\Basegroup(1 - \beta\Basegroup)}$ or $\nicefrac12$, the threshold (including a Bonferroni correction) is 
\begin{equation}
    \label{eq:ztest_thresh}
    \theta_t(\alpha) := t \cdot \beta\Basegroup + C \sqrt{
    2.07 t
    \ln\left(|\Groups| \frac{(2 + \log_2(t))^2}{\alpha}\right)
    }.
\end{equation}

\paragraph{Theoretical guarantees.}
Our first guarantee is a bound on the probability that any group is incorrectly flagged. 
\begin{restatable}[Validity]{theorem}{ztestvalidity}
\label{thm:validity_ztest}
Running Algorithm~\ref{alg:abstract} with $\theta_t(\alpha)$ as in Equation \eqref{eq:ztest_thresh}, setting $C = \nicefrac
12$, and $\Logwealth_t^G$ updated as in Equation~\eqref{eq:ztest_update}, guarantees that 
the probability that $\FlagG$ will ever contain a group $G$ where $\NullH$ is true is at most $\alpha$, i.e. 
\[
\Pr\left[\exists t: \exists G \in \FlagG \text{ s.t. } \NullH \text{ holds}\right] \leq \alpha.
\]
\end{restatable}

The choice of $C$ affects the nature of the guarantee: the true, finite-sample anytime-validity guarantee requires $C = \nicefrac12$. If instead $C = \sqrt{\beta\Basegroup(1 - \beta\Basegroup)}$, then, strictly speaking, the guarantee holds only asymptotically. However, a higher value of $C$ affects stopping time unfavorably, so the asymptotic approximation can be useful practically. In this case, care must be taken to ensure that the algorithm does not erroneously reject too early due to noise; one way to implement this is to mandate a minimum stopping time. 

Finally, we give a stopping time guarantee for this test.
\begin{restatable}[Power]{theorem}{ztestpower}
\label{thm:power_ztest}
Let $T$ be the stopping time of 
Algorithm~\ref{alg:abstract} with $\theta_t(\alpha)$ as in Equation \eqref{eq:ztest_thresh}, $C = \nicefrac
12$, and $\Logwealth_t^G$ as in Equation~\eqref{eq:ztest_update}.
Let $\Delta_{\max} = \max_{G \in \Groups} \mu_G - \beta\Basegroup.$
If $\Delta_{\max} > 0$, then $\Pr[T < \infty] = 1$. 
Furthermore, with probability $1 - \nicebonf$, we have $T  \leq \widetilde{\mathcal{O}} \left( \frac{\ln(|\Groups|) + \ln(1/\alpha)}{\Delta_{\max}^2}\right)$, and for any $\delta \in (0, \nicebonf)$, we have with probability at least $1 - \delta$ that
$T  \leq \widetilde{\mathcal{O}} \left( \frac{\ln(1/\delta)}{\Delta_{\max}^2}\right)$.
\end{restatable}
\subsection{Betting-style approach}
\label{subsec:e-val-alg}
We refer to our second algorithm as a \textit{betting-style} approach, due to the way we construct our test statistics \citep{shafer2021testing, waudby2024estimating, vovk2021values, chugg2024auditing}; one way to interpret this approach is that the test ``bets against'' the null hypothesis $\NullH$ being true. We direct the reader to these references for more detailed technical exposition\iftoggle{icml}{.}{ and connection with literature on martingales, gambling, and finance. For us, these methods provide an adaptive algorithm which find a middle ground between the two approaches in the previous section: the betting-style approach achieves finite-sample validity but empirically terminates quickly when the null is false.}

\paragraph{Updating the test statistic $\Logwealth_t^G$.}
As in the previous approach, we let $\Logwealth_t^G$ represent some accumulated amount of evidence against the null hypothesis $\NullH$ by time $t$, with a higher value of $\Logwealth_t^G$ corresponding to greater level of evidence.\footnote{The quantity $\exp(\Logwealth_t^G)$ can also be referred to as an \textit{e-value} \citep{vovk2021values}, a measure of evidence against a null hypothesis similar to a p-value.} We initialize $\Logwealth_0^G = 0$, and use the update rule 
\begin{equation}\label{eq:wealth_update}
\Logwealth_t^G \leftarrow \Logwealth_{t-1}^G + \ln\left(1 + \lambda_t^G (\mathbf{1}_{X_t\in G} - \beta \mu_G^0)\right),
\end{equation}
with $\lambda_1^G, \ldots, \lambda_t^G \in [0, 1]$. 
\iftoggle{icml}{}{Note that this expression is similar to the running sum used in Section~\ref{subsec:ztest}.}
Here, the algorithm accumulates a nonlinear function, with an adaptive parameter $\lambda_t^G$ that weights the influence of each new sample. 
Our setting of $\lambda_t$ is motivated by the goal of minimizing stopping time under the alternative, and thus to maximize $\Logwealth_t^G$. 
\iftoggle{icml}{}{Taking $\lambda_{t} = 0$ means $\Logwealth_t^G$ remains the same regardless of what new information is received at time $t$. On the other hand, $\lambda_{t} = \nicefrac{1}{\beta\mu_G^0}$ means that if we receive evidence in accordance with $\NullH$ then $\Logwealth_t^G$ will decrease substantially; but, if we instead receive evidence \textit{against} the null, i.e. $X_{t}\in G$, we maximally increase $\Logwealth_{t}^G$.}
Drawing from the well-studied problem of portfolio optimization in the online learning literature \citep{cover1991universal,
zinkevich2003online, hazan2016introduction}, we use Online Newton Step \citep{hazan2007logarithmic, cutkosky2018black} to ensure that $\Logwealth_t^G$ is not too far from the best achievable in hindsight. This results in the following update for $\{\lambda_t\}_{t \geq 1}$:
\begin{equation}\label{eq:bet_update}
\lambda_{t+1}^G \gets \mathop{\text{Proj}}\limits_{\left[0, 1\right]}\left(\lambda_t^G + \tfrac{2}{2- \ln(3)}\cdot\tfrac{z_t}{1 + \sum_{s \in [t]}z_{s}^2}
   \right),
\end{equation}
where $z_t = \frac{\1[X_t\in G ]- \beta \Basegroup}{1 + \lambda_t^G(\1[X_t\in G]- \beta \Basegroup)}$, and $\lambda_0 = 0$.\footnote{The constant $\frac
{2}{2 - \ln(3)}$ is due to \citet{cutkosky2018black}, who give a tighter version of ONS than in \citet{hazan2007logarithmic}.}

\paragraph{Setting the threshold $\theta_t(\alpha)$.}
Unlike the sequential Z-test, we use the same threshold for all timesteps. Including a Bonferroni correction, we use $\theta_t(\alpha) :=   \ln(\nicefrac{|\Groups|}{\alpha})$ for all $t$; the motivation for this setting will become clear in our discussion of Theorem \ref{thm:validity_evals}.

\paragraph{Theoretical guarantees.}
We first give a validity guarantee that is essentially identical to the Sequential Z-test. 

\begin{restatable}[Validity]{theorem}{evalsvalidity}\label{thm:validity_evals} Running Algorithm~\ref{alg:abstract} with $\theta_t(\alpha) = \ln{(\nicefrac{|\Groups|}{\alpha})}$ and $\Logwealth_t^G$ updated as per Equations~\eqref{eq:wealth_update} and \eqref{eq:bet_update} guarantees that 
the probability that $\FlagG$ will ever contains a group $G$ where $\NullH$ is true is at most $\alpha$, i.e. 
\[
\Pr\left[\exists t: \exists G \in \FlagG \text{ s.t. } \NullH \text{ holds}\right] \leq \alpha.
\]
\end{restatable}
This result follows directly from the prior work referenced at the beginning of this section. At a high level, every sequence $\{\exp(\Logwealth_t^G)\}_{t \geq 1}$ is a non-negative super-martingale under $\NullH$; informally, this means that under the null hypothesis, the sequence $\{\exp(\Logwealth_t^G)\}_{t \geq 1}$ should be non-increasing, in expectation. 
This allows us to apply Ville's inequality, which guarantees that it is unlikely that $\exp(\Logwealth_t^G)$ ever becomes too large under $\NullH$. More specifically, for any $\alpha \in (0,1)$, under the null, $\Pr[\exists t: \exp(\Logwealth_t^G) > 1/\alpha] \leq \alpha $. Thus, maintaining a threshold of $\theta_t(\alpha) = \ln(\nicefrac{|\Groups|}{\alpha})$ is sufficient to provide a per-hypothesis $\nicebonf$-validity guarantee, and thus $\alpha$-validity overall. 

We also provide the following bound on stopping time; see Appendix \ref{app:eval} for additional discussion of the $\Logwealth_\star$ notion of gap. 
\begin{restatable}[Power]{theorem}{evalspower} 
\label{thm:power_evals}
Let $T$ be the stopping time of Algorithm~\ref{alg:abstract} with $\theta_t(\alpha) = \ln{(|\Groups|/\alpha)}$ and $\Logwealth_t^G$ updated as per Equations~\eqref{eq:wealth_update} and \eqref{eq:bet_update}. If $\max_{G \in \Groups} \mu_G - \beta\Basegroup > 0$, then, we have that $\Pr[T < \infty] = 1$. Furthermore, 
\[
\E[T] \leq \mathcal{O}\left(\frac{1}{\Logwealth_\star^2} + \frac{\ln(|\Groups|) + \ln(1/\alpha)}{\Logwealth_\star}\right)
\]
where $\Logwealth_\star := \max_{G\in\mathcal{G}, \lambda\in[0,1]}\E[\ln(1+ \lambda(\mathbf{1}_{X_t\in G} -\beta\Basegroup))]$ is the maximal expected one-step increase in $\omega_t^G$ over all groups and choices of $\lambda$.
\end{restatable}
We conclude this section with two further remarks on Theorems \ref{thm:power_ztest} and \ref{thm:power_evals} in the context of our test. First, our modeling in Section~\ref{sec:interpretation} measures severity of harm via a \textit{multiplicative} factor of overrepresentation.
However, both notions of gap in Theorems \ref{thm:power_ztest} and \ref{thm:power_evals} also on the absolute size of the group $\mu_G$.
Thus, for two groups $G$ and $G'$ with identical multiplicative gaps, i.e. $\nicefrac{\mu_G}{\Basegroup} = \nicefrac{\mu_{G'}}{\mu_{G'}^0}$, the test would stop faster in expectation for $G$ if and only if $\Basegroup > \mu_{G'}^0$. That is, if two groups are ``harmed'' to the same extent, both algorithms will identify the larger one first. 

Second, for both tests, the Bonferroni correction results in only an additive factor ($\nicefrac{\ln(|\Groups|)}{\Delta_{\max}^2}$ in Theorem \ref{thm:validity_ztest}, and $\nicefrac{\ln(|\Groups|)}{\omega_\star}$ in Theorem \ref{thm:validity_evals}) in stopping time over the scenario where we had only been testing the one group with the largest gap. 
This means that, in terms of worst-case guarantee on stopping time, the contribution of the Bonferroni correction is small relative to the contribution of the test level $\alpha$ and, especially, to the gap. In fact, the impact of Bonferroni on real-world data appears to be much smaller even than this additive term.

%% file: 5-experiments.tex
\section{Real-World Examples}
\label{sec:experiments}

To demonstrate the applicability of our approach, we apply our framework to two real-world datasets. 
We begin by showing that our approach correctly and quickly identifies that young men experience myocarditis after the COVID-19 vaccine; then, on mortgage allocation data, we show that we identify known instances of discrimination under several reasonable reporting models. 
\iftoggle{icml}{\jessedit{Code for all experiments, including instructions for data download and pre-processing, is available in the supplemental materials; additional experimental details can be found in Appendix \ref{app:expts}.}}{
Code for all experiments, including instructions for data download and pre-processing, is available online \href{https://github.com/jessica-dai/reporting}{here}.
}

\subsection{Myocarditis from COVID-19 vaccines} 

\begin{table*}[t!]
\centering
  \begin{tabular}{l|ll|ll|ll}
    \toprule
    \multirow{2}{*}{} &
    \multicolumn{2}{c|}{\textit{Asymptotic Z-test}} &
    \multicolumn{2}{c|}{\textit{Finite-sample Z-test}} &
    \multicolumn{2}{c}{\textit{Betting-style test}}~\\
      & {\small{(M, 18-29)}} & {\small{(M, 12-17)}} &{\small{(M, 18-29)}} & {\small{(M, 12-17)}} & {\small{(M, 18-29)}} & {\small{(M, 12-17)}}  \\
      \hline\hline
    $\beta = 2.0$ & 34 (Feb. 22) & 256 (May 10) & 69 (Mar. 28) & 530 (May 30) & 61 (Mar. 23) &  241 (May 8) \\
    $\beta = 2.5$ & 49 (Mar. 10) & 302 (May 15) & 74 (Mar. 31) & 546 (Jun. 1) & 69 (Mar. 28) & 259 (May 11) \\
    $\beta = 3.0$ & 70 (Mar. 30) & 324 (May 18) & 111 (Apr. 20) & 612 (Jun. 6) & 80 (Apr. 5) & 302 (May 15) \\
    \bottomrule
  \end{tabular}
  \vspace{1em}
  \caption{On real historical sequence of myocarditis reports, time to identification of harmed groups. In each cell, we report the number of total reports to the rejection of the hypothesis corresponding to (M, 18-29) and the number of total reports corresponding to (M, 12-17). In all tests, the (M, 18-29) group is identified first---vaccines were authorized for the 12-15 age group only in May. }
  \label{table:covid}
\end{table*}

It is by now well-known that COVID-19 vaccines induce an elevated risk of myocarditis among young men. While initial suspicions of elevated myocarditis risk relied on case studies (e.g., \citet{mouch2021myocarditis, larson2021myocarditis,marshall2021symptomatic}), a more systematic understanding---including the pattern of disproportionate impact---was made possible by post-hoc analysis of reports from incident databases. \citet{barda2021safety} appears to be the first analysis based on a database of reports, but did not disaggregate by demographic subgroups; the confirmation of young men as the most drastically-impacted group came in later studies (e.g., \citet{witberg2021myocarditis, oster2022myocarditis}).

In the U.S., these reports are collected inthe  Vaccine Adverse Event Reporting System (VAERS).
If we had been able to run the hypothesis tests proposed in the preceding sections on the reports collected in VAERS, would we have correctly identified this problem---and if so, how quickly? 
Concretely, we let $Y_i$ be the event that individual $i$ experiences myocarditis after receiving a COVID-19 vaccine, and run the test with the end-goal of identifying elevated incidence rate $\Pr[Y_i \mid X_i \in G]$ for group(s) $G$ corresponding to adolescent men (ages 12-17 and 18-29). 

\iftoggle{icml}{}{
\paragraph{Data sources.}
The Vaccine Adverse Event Reporting System (VAERS) is a national adverse event incident reporting database for U.S.-licensed vaccines, co-managed by the Centers for Disease Control and Prevention (CDC) and the U.S. Food and Drug Administration (FDA)~\citep{chen1994vaccine,shimabukuro2015safety}. 
The database is a combination of voluntary reports from patients that have received the vaccine, as well as mandatory reports from vaccine manufacturers and healthcare professionals. 
For this case study, we filter the set of publicly-available reports from VAERS to reports about the COVID-19 vaccine with a complaint keyword including ``myocarditis.'' As for how a database administrator would have known to focus on myocarditis \textit{a priori}, one might imagine, for example, that the series of case studies found in early 2021 raised the alarm that more systematic analysis was warranted for myocarditis in particular.

To determine per-demographic base rates, i.e. to compute $\{\Basegroup\}_{G\in\Groups}$, we utilize 
VaxView, a government dataset tracking national vaccine coverage, managed by the CDC. VaxView does not track vaccination rates by granular subgroups, only providing coverage rates disaggregated by age, gender, and ethnicity separately. We thus impute the vaccination rates for intersections of subgroups (e.g., ``12-17, M'') by multiplying the known marginal rates (i.e., $\mu_{(12-17, M)}^0 := \mu_{(12-17)}^0 \cdot \mu_{(M)}^0$).}

\paragraph{Defining $\Groups$.} 
We consider (intersections of) sex and age buckets to be the subgroups of interest.\footnote{While in principle it would have been interesting to also consider race/ethnicity, we are limited by the availability (and granularity) of the data given in VAERS, which does not include information on ethnicity/race in reports.}  Age buckets are discretized into 0-4, 5-11, 12-17, 18-29, 30-39, 40-49, 50-64, 65-74, and 75+; the sex categories represented in the data are (binary) male and female. After removing groups for which no vaccines were recorded, $\Groups$ contains 29 groups. 

\paragraph{Setting $\beta$.}
For this application, absolute incidence rate (that is, $\Pr[Y = 1 \mid G]$) is the quantity of interest to use for determining $\beta$.
As suggested by Proposition \ref{prop:reporting-conversion}, setting $\beta$ requires considering three quantities of interest: the threshold on an ``unacceptable'' incidence rate, the relative rates of true reporting $\Truereportrate$, and the relative rates of false reporting $\Falsereportrate$. Then, we can set $\beta = \max_G \left((\Truereportrate-\Falsereportrate)\cdot \mathrm{IR} + \Falsereportrate\right)$. 

We will choose 0 as the threshold on an ``unacceptable'' incidence rate.\iftoggle{icml}{}{\footnote{One might follow existing practice and use the per-group expected rate of myocarditis to benchmark an unacceptable incidence rate (e.g. as provided in Table 2 of \citet{oster2022myocarditis}, which suggests at most 2 cases per million doses). 
However, in addition to this expected incidence rate being very small (and, for any practical purposes, being vastly dominated by the other reporting terms), it also implicitly relies on reports so that the benchmark quantities are  $r \cdot \mathrm{IR}$, rather than just $\mathrm{IR}$, and thus depend on the unknown reporting rate $r$.}}
It is therefore sufficient to set $\beta = \max_G(\Falsereportrate)$. While this is quantity cannot be determined from report data alone, a conservative assumption could be that any group erroneously reports at most twice the average reporting rate  over the population, with $\Falsereportrate = 2.0$. 
If the algorithm is first run with $\beta=2.0$, stopping and flagging a group very quickly, the test may be re-run with increasing values of $\beta$, as a higher $\beta$ corresponds to a more severe true incidence rate; thus, we also show results for $\beta=2.5$ and $\beta=3$.\footnote{Re-using this data is statistically valid due to the equivalence between one-sided hypothesis testing and confidence sequences.}

\paragraph{Results.}  
We begin by running our algorithms on the actual sequence of reports in chronological order, as received in VAERS. In particular, we consider Algorithm \ref{alg:abstract} instantiated with $\Logwealth_t^G$ updated according to Equation \eqref{eq:ztest_update} and $\theta_t(\alpha)$ as in \eqref{eq:ztest_thresh} and $C = \nicefrac12$ (\textit{Finite-sample Z-test}); with $\Logwealth_t^G$ updated according to Equation \eqref{eq:ztest_update} and $\theta_t(\alpha)$ as in \eqref{eq:ztest_thresh} and $C = \sqrt{\beta\Basegroup(1 - \beta\Basegroup)}$ (\textit{Asymptotic Z-test}); and with $\Logwealth_t^G$ updated according to Equations \eqref{eq:wealth_update} and \eqref{eq:bet_update}, and $\theta_t(\alpha) = \ln (\nicefrac{|\Groups|}{\alpha})$ 
(\textit{Betting-style test}). 
For the asymptotically-valid Z-test, we require a minimum stopping time of $t = 25$, to prevent early rejections.
We run all experiments for $\alpha = 0.1$. 

In Table \ref{table:covid}, we report the stopping time---that is, the number of reports it takes for the first null hypothesis to be rejected---of each algorithm for various values of $\beta$, as well as the corresponding date by which an alarm would have been triggered. Note that, in all tests, the (M, 18-19) group is identified first. This is consistent with the timeline of regulatory approvals: vaccines were authorized for ages 12-15 only by May 10 \citep{lovelace2021fda}.

To explore the robustness of these results, we also run synthetic experiments, 
permuting the ordering of reports to get a sense of possible variance in the stopping time. We run 100 random permutations of the full set of reports. 
Figure \ref{fig:covid-all-beta2} tracks the number of reports it takes for each algorithm to reject the null hypothesis for any group---that is, a scenario when the test is stopped and an alarm is raised as soon as one harmed group is identified. Each point on these plots reflects the number of trials (out of 100) in which a rejection has occurred by time $t$, when tests are run at $\beta=2$. 

In Figure \ref{fig:covid-all-beta2}, we compare the performance of the three algorithms.\iftoggle{icml}{\footnote{To interpret the figure, by time $t=100$, the asymptotically-valid z-test had stopped (identified harm) in all 100 permutations; the betting-style test stopped in around 80 permutations; and the finite-sample z-test stopped in around 20 permutations.}}{ To interpret the figure, by time $t=100$, the asymptotically-valid z-test had already identified harm in all 100 permutations; the betting-style test identified harm in around 80 permutations; and the finite-sample z-test had only identified harm in around 20 permutations.}
Figure \ref{fig:covid-all-beta2} shows a clear ordering in terms of how quickly each algorithm tends to identify harm: the asymptotically-valid sequential z-test (dashed, red) is faster than the betting-style algorithm (solid, purple), which is faster than the finite-sample z-test (dotted, yellow). 

\iftoggle{icml}{
\begin{figure}[h]
\centering
\begin{minipage}[c]{0.45\textwidth}
    \includegraphics[width=0.95\linewidth]{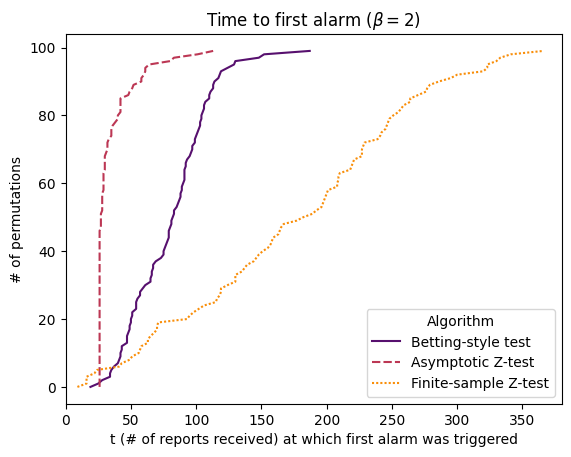}
\end{minipage}
    \caption{\small  Stopping time (i.e. first identification of harm) for each algorithm, over 100 random permutations of COVID-19 vaccine reports, with $\beta=2$. Each point on the plot reflects the number of trials (out of 100) in which a rejection has occured by time $t$. 
    }
    \label{fig:covid-all-beta2}
\end{figure}
}{
\begin{figure*}[t!]
\hfill
\begin{minipage}[c]{0.55\textwidth}
    \includegraphics[width=0.95\linewidth]{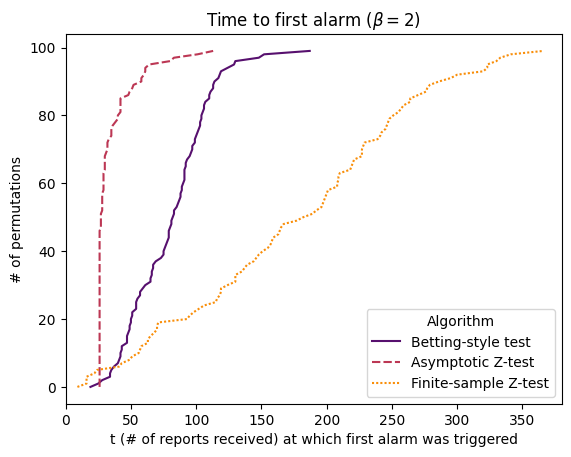}
\end{minipage}
    \hfill
\begin{minipage}[c]{0.35\textwidth}
    \caption{\small  Number of reports ($t$) it takes for each algorithm to reject the null hypothesis for any group (i.e. first identification of harm), over 100 random permutations of COVID-19 vaccine report database. Tests are run with $\beta=2$. Each point on the plot reflects the number of trials (out of 100) in which a rejection has occured by time $t$. 
    }
    \label{fig:covid-all-beta2}
\end{minipage}
\hfill
\end{figure*}
}

\iftoggle{icml}{}{
We also explore the impact of Bonferroni correction for multiple hypothesis testing on stopping time.
In  Figure \ref{fig:covid-beta2-bonf}, we show the same axes---number of reports to first alarm on the x-axis, vs. number of permutations in which an alarm was triggered on the y-axis---for the three algorithms at $\beta=2$. As expected, the invalid version of the test, which has a lower threshold for rejecting each null, stops more quickly for all three algorithms (dashed, lighter). 
The difference between the invalid version and the valid version (solid, darker) is relatively minor, though the impact varies across algorithms. 
% For the betting-style test, recall that the stopping-time upper bound in Theorem \ref{thm:power_evals} suggested that the impact of adding the Bonferroni correction was an additive factor of $\mathcal{O}(\nicefrac{\log(|\Groups|)}{\Delta_{\mathrm{max}}^2})$. 
% In the plot for the betting-style test, we add a line (dotted, lightest) that adds exactly $\lceil\nicefrac{\log(|\Groups|)}{\Delta_{\mathrm{max}}^2}\rceil$  to the actual stopping time of the algorithm without the Bonferroni correction; note that the actual stopping time of the betting-style algorithm with the multiple testing correction is much faster than what is suggested by the theoretical upper-bound.

\begin{figure*}[h]
\hfill
\begin{minipage}[c]{0.32\textwidth}
    \includegraphics[width=0.97\linewidth]{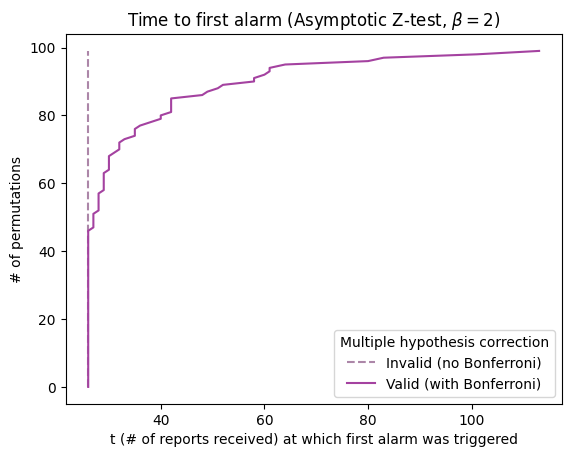}
\end{minipage}
\hfill
\begin{minipage}[c]{0.32\textwidth}
    \includegraphics[width=0.97\linewidth]{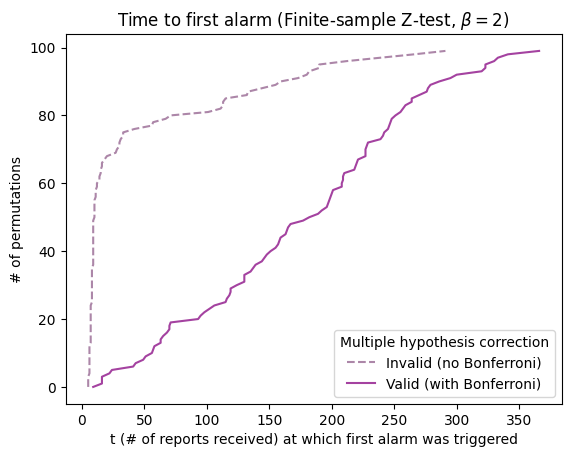}
\end{minipage}
\hfill
\begin{minipage}[c]{0.32\textwidth}
    \includegraphics[width=0.97\linewidth]{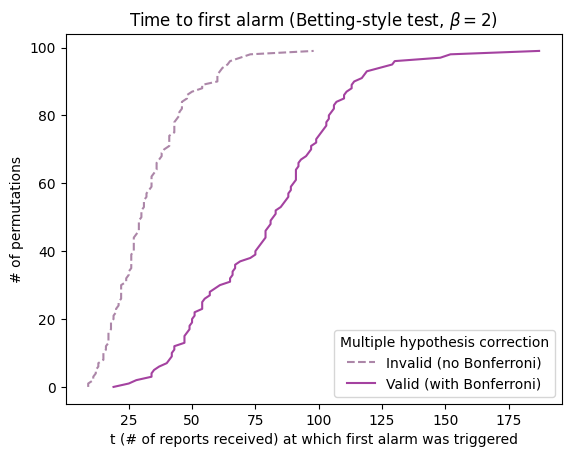}
\end{minipage}
\hfill
    \caption{\small Impact of multiple hypothesis correction on stopping time across algorithms. As in Figure \ref{fig:covid-all-beta2}, each point on the plot reflects the number of trials (out of 100) in which a rejection has occurred by time $t$. 
    In all plots, the lighter, dashed line reflects stopping time of the invalid test that does not correct for multiple testing; the dark, solid line reflects stopping time of the valid test including a Bonferroni correction. 
    }
    \label{fig:covid-beta2-bonf}
\end{figure*}
}

Overall, our experimental results suggest that our proposed tests would in fact have been effective in determining that young men were disproportionately affected by myocarditis. Moreover, though it is difficult to determine exact timelines and the nature of clinical practice during early phases of the vaccine rollout, it is possible that such a test could have identified problems using less data---that is, more quickly---than was actually used for this finding.

\subsection{Mortgage Allocations}

In 2021, \citet{martinez2021secret} found that, based on publicly-released data from the Home Mortgage Disclosure Act (HMDA), 
substantial racial disparities in 2019 loan approvals persisted even after controlling for financial status of applicants---most notably, healthy debt-to-income ratios (DTI). 
If loan applicants had been able to submit reports when they believed they had experienced unfavorable outcomes, could those reports have been used to identify this discrimination? If so, how accurately, and how quickly? 

We are interested primarily in disparity among applicants with healthy DTI, even though all loan applicants would have been eligible to submit reports. 
Concretely, we let $A_i = 0$ be the event that a loan is not made to applicant $i$, and $Z_i = 1$ be the event that applicant $i$ has a healthy debt-to-income ratio. Then, we let $Y_i = \{A_i = 0, Z_i = 1\}$ be the event that individual $i$ has a healthy DTI and did not receive a loan, and run the test with the end-goal of identifying groups that have relatively high rates of loan denials for applicants with healthy DTI, i.e. $\frac{\Pr[A_i = 0, Z_i = 1 \mid X_i \in G]}{\Pr[A_i = 0, Z_i = 1]}$. 

\iftoggle{icml}{}{
\paragraph{Data sources.}
We use the data (and preprocessing code) of \citet{martinez2021secret}, which uses 2019 data from the HMDA.\footnote{The Consumer Financial Protection Bureau (CFPB) collects and publishes this data from financial institutions annually, with a two-year lag; the report (and our work) uses 2019 data which is finalized as of Dec. 31, 2022. The most recent year for which data is available is 2022, though it is available for edits through 2025.}
The analysis of \citet{martinez2021secret} used the full year of data from 2019; we reduce the dataset to applications for conventional loans at three of the largest lending institutions, from applicants who have positive income. We assume that reports will only come from applicants whose loans were denied; in all, there are 183k applicants which satisfy these criteria. 
}

\paragraph{Defining $\Groups$.} 
While \citet{martinez2021secret} analyzed disparities with respect to race, we define groups as all possible intersections of demographic features across gender, race, and age. 
The available race categories include Native, Asian, Black, Pacific Islander, White, and Latino; sex categories include female, male, and unknown/nonbinary; and age categories include $<$25, 25-34, 35-44, 45-54, 55-64, and 65+. 
In total, after removing groups which comprise less than 0.1\% of all loan applicants, $\Groups$ contains 115 groups. 

\paragraph{Setting $\beta$.}
In this application, the quantity of interest is relative risk, so we draw on Proposition \ref{prop:relativerisk-conversion} to inform our setting of $\beta$. We will set our relative risk threshold to be 1.2---that is, we want our algorithm to raise an alarm when any group experiences event $Y$ 20\% more frequently than average over the population. 
Recall Definition \ref{def:rir} and Proposition \ref{prop:relativerisk-conversion}: to flag relative risk at 1.2, $\beta$ should be set to $1.2 \cdot b$ where $b = \max_G \nicefrac{\rho_G}{\rho}$, with $\rho_G = \frac{\Pr[R =1 \mid G]}{\Pr[Y = 1 \mid G]}$ and $\rho = \frac{\Pr[R = 1]}{\Pr[Y = 1]}$; that is, $b$ is the extent to which the group-conditional report-to-incidence ratio for any group deviates from the population average report-to-incidence ratio. 

As before, we can first test at $\beta = 1.2$ , then re-test for higher values of $\beta$; in this case, we will also test $\beta = \{1.4, 1.6, 1.8\}$. Setting $\beta = 1.2$ corresponds to assuming $b = 1$, i.e., no variance in report-to-incidence ratios across groups; the additional values of $\beta$ suggest possible values of $b = \nicefrac{7}{6}, \nicefrac{4}{3},$ and $\nicefrac{3}{2}$, respectively. 

\paragraph{Reporting models.}
The existence of verifiable disparities in this dataset allows us to evaluate the efficacy of our methods under varying models of reporting---that is, whether 
our algorithms identify groups that do in fact have high rates of healthy DTI denials, even if it is not the case that every report $X_i$ corresponds to $\Badevent_i$ actually occurring. 
\iftoggle{icml}{
Modeling the idea that reporting behavior may be related to financial health, we simulate the following possible patterns of reporting.
In the \textit{Correlated} model, the likelihood of reporting increases with financial health; in the \textit{All Denials} model, all denials submit reports regardless of financial health; in the \textit{Anti-Correlated} model, individuals with worse financial health are more likely to report (see Appendix \ref{app:expts}). 
}{
The dataset gives several levels of financial health as measured by DTI---in ascending order, are``Struggling'', ``Unmanageable,'' ``Manageable,'' and ``Healthy.'' 
Modeling the idea that reporting behavior may be related to financial health, we use these categories to simulate the following possible patterns of reporting.
\begin{enumerate}[(1)]
    \item \textit{Correlated:} The likelihood of reporting increases with financial health. That is, ``Healthy'' denials report with probability 0.9, ``Manageable'' with probability 0.5, ``Unmanageable'' with probability 0.3, and ``Struggling'' with probability 0.1. Under this reporting model, the 95th-percentile (among all groups) $\nicefrac{\rho_G}{\rho}$ is 1.2, and $\max_G \nicefrac{\rho_G}{\rho} = 1.4.$
    \item \textit{All Denials:} All denials submit reports regardless of financial health. Under this reporting model, the 95th-percentile $\nicefrac{\rho_G}{\rho}$ is 1.5, and $\max_G \nicefrac{\rho_G}{\rho} = 2.3.$
    \item \textit{Anti-Correlated:} The (unlikely) case where individuals with worse financial health are more likely to report, i.e. ``Healthy'' denials report with probability 0.1, ``Manageable'' with probability 0.5, ``Unmanageable'' with probability 0.7, and ``Struggling'' with probability 0.9. Under this reporting model, the 95th-percentile $\nicefrac{\rho_G}{\rho}$ is 1.8, and $\max_G \nicefrac{\rho_G}{\rho} = 2.7.$
\end{enumerate}
Note that the ground-truth ratios $\nicefrac{\rho_G}{\rho}$ would have been unknown at the time that a practitioner sets $\beta$; we are able to determine these only because we have full information about the dataset and control over the reporting model. However, these computations suggest that the assumptions on reporting rates implied by the settings of $\beta = \{1.2, 1.4, 1.6, 1.8\}$ are generally reasonable, especially after considering outliers---note the disparity between the 95th-percentile vs max ratios of $\nicefrac{\rho_G}{\rho}$, especially for the \textit{All Denials} and \textit{Anti-Correlated} models.
}

\iftoggle{icml}{
\begin{table*}[t!]
\centering
  \begin{tabular}{c|cc|cc|cc}
    \toprule
    \multirow{2}{*}{} &
    \multicolumn{2}{c|}{\textit{Asymptotic Z-test}} &
    \multicolumn{2}{c|}{\textit{Finite-sample Z-test}} &
    \multicolumn{2}{c}{\textit{Betting-style test}}~\\
      & {\small{Stopping time}} & {\small{Relative risk}} &{\small{Stopping time}} & {\small{Relative risk}} &{\small{Stopping time}} & {\small{Relative risk}}   \\
      \hline\hline
    \textit{Correlated} & 886 & 1.68 & 11755 & 1.73 & 4157 & 1.82 \\
    \textit{All Denials} &  586 & 1.69  & 7410 & 1.72 & 2714 & 1.75 \\
    \textit{Anti-Corr.} & 271 & 1.05 & 4668 &  1.72  & 1688 &  1.71 \\
    \bottomrule
  \end{tabular}
  \vspace{1em}
  \caption{Average stopping times (i.e. time to first alarm) and true relative risk (i.e., $\frac{\Pr[A_i = 0, Z_i = 1 \mid X_i \in G]}{\Pr[A_i = 0, Z_i = 1]}$) of first-identified group over 100 random permutations, for $\beta=1.6$, across algorithms and reporting models.}
  \label{table:hmda}
\end{table*}
}{
\begin{table*}[t!]
\centering
\begin{tabular}{c|c|cc|cc|cc}
  \toprule
  \multirow{2}{*}{} & \multirow{2}{*}{\small{Reporting model}} &
  \multicolumn{2}{c|}{\textit{Asymptotic Z-test}} &
  \multicolumn{2}{c|}{\textit{Finite-sample Z-test}} &
  \multicolumn{2}{c}{\textit{Betting-style test}}~\\
  & & {\small{Stopping time}} & {\small{Rel. risk}} & {\small{Stopping time}} & {\small{Rel. risk}} & {\small{Stopping time}} & {\small{Rel. risk}} \\
  \hline\hline
  \multirow{3}{*}{$\beta=1.2$} & \textit{Correlated} & 85 & 1.62 & 2002 & 1.67 & 638 & 1.70 \\
  & \textit{All Denials} & 69 & 1.59 & 1546& 1.60 & 519 & 1.65 \\
  & \textit{Anti-Corr.} & 60 & 1.50 & 1065 & 1.53 & 403 & 1.65 \\
  \hline
    \multirow{3}{*}{$\beta=1.4$} & \textit{Correlated} & 316 & 1.69 & 4306 & 1.73 & 1542 & 1.77 \\
  & \textit{All Denials} & 163 & 1.62 & 3214 & 1.72 & 1073 & 1.72 \\
  & \textit{Anti-Corr.} & 95 & 1.47 & 2215 & 1.66 & 718 & 1.68 \\
  \hline
    \multirow{3}{*}{$\beta=1.6$} & \textit{Correlated} & 886 & 1.68 & 11755 & 1.73 & 4157 & 1.82 \\
  & \textit{All Denials} & 586 & 1.69 & 7410 & 1.72 & 2714 & 1.75 \\
  & \textit{Anti-Corr.} & 271 & 1.05 & 4668 & 1.72 & 1688 & 1.71 \\
  \hline
    \multirow{3}{*}{$\beta=1.8$} & \textit{Correlated} & 4959 & 1.74 & ---$^1$ & --- & 16425$^2$ & 1.98 \\
  & \textit{All Denials} & 2703 & 1.72 & 29751$^3$ & 1.73 & 9977 & 1.89 \\
  & \textit{Anti-Corr.} & 935 & 1.58 & 14072 & 1.73 & 4629 & 1.76 \\
  \hline
\end{tabular}
\vspace{1em}
\caption{Average stopping times (i.e. time to first alarm) and true relative risk (i.e., $\frac{\Pr[A_i = 0, Z_i = 1 \mid X_i \in G]}{\Pr[A_i = 0, Z_i = 1]}$) of first-identified group over 100 random permutations, for varying $\beta$, across algorithms and reporting models.
For $\beta=1.8$, some combinations of algorithm/reporting model failed to stop within 40,000 steps for some trials: 
$^1$stopped in 0/100 trials, $^2$stopped in 99/100 trials, $^3$stopped in 76/100 trials.}
\label{table:hmda}
\end{table*}
}

\paragraph{Results.} We run all three algorithms discussed in Section \ref{sec:algs} at $\alpha=0.1$, for all four reporting models discussed above, and for $\beta=\{1.2, 1.4, 1.6, 1.8\}$. For the asymptotically-valid Z-test, we (heuristically) choose a higher minimum stopping time of 50, to reflect the more challenging problem instance compared to the vaccine reporting problem.
For each algorithm, reporting model, and $\beta$, we again run 100 random permutations.\iftoggle{icml}{\footnote{Since we are simulating reporting, there is no ``true'' historical sequence of reports to run an algorithm on, unlike in Table \ref{table:covid}.}}{ (Since we are simulating reporting, there is no ``true'' historical sequence of reports to run an algorithm on, unlike in Table \ref{table:covid}.)}

One important question for this application is the extent to which our tests identify the type of harm we are interested in, across various reporting models: while the algorithms guarantee statistical validity in terms of overrepresentation (i.e., in terms of whether $\mu_G \geq \beta\Basegroup$), they cannot intrinsically guarantee that reports themselves reflect true harm.
With the benefit of hindsight (and access to the full dataset), we are able to calculate ``ground truth'' relative risks; the hope for our algorithms is that they identify groups that actually do experience elevated relative risk.

Our results suggest that this is indeed generally the case, although the actual behavior varies by algorithm and reporting model. Table \ref{table:hmda} shows report the average stopping times and average true relative risks of the first-identified group for 100 permutations\iftoggle{icml}{, for $\beta = 1.6$ (see Appendix \ref{app:expts} for results for all $\beta$)}{}. 
Across all algorithms\iftoggle{icml}{}{ and values of $\beta$}, the stopping time under the \textit{Correlated} reporting model is the longest, followed by the \textit{All Denials} and \textit{Anti-Correlated} reporting models. On the other hand, the relative risk of the group that is first identified in each of these settings follows the same ordering, with the \textit{Correlated} model having the highest relative risk. That is, more ``favorable'' reporting behavior required a test to run longer, but the group identified is more severely harmed, whereas more ``adversarial'' reporting behavior raised an alarm sooner, but identified a less severely-harmed group. 

Similar tradeoffs arise when comparing algorithms: the asymptotically-valid Z-test stops far more quickly, but appears to identify less severely-harmed groups. On the other hand, while the betting-style test and the finite-sample Z-test tend to identify similarly-harmed groups, the latter stops much faster than the former; overall, it appears that the betting-style test is a reasonable approach to balancing fast identification with confidence in the severity of harm. 

While overall trends across algorithms and reporting models are consistent across values of $\beta$, seeing these results for different $\beta$ highlights an additional insight. 
While it is to be expected that stopping times (and the ground-truth relative risks) should increase with $\beta$, the increase in stopping time is dramatic---by sometimes by orders of magnitude---even for what appear to be relatively small changes in $\beta$. Moreover, the \textit{disparity} in stopping time across reporting models also increases dramatically with $\beta$. In fact, for $\beta=1.8$, some combinations of reporting and algorithm do not stop within 40,000 steps in at least one trial. 

\iftoggle{icml}{}{
\paragraph{Practical takeaways.}
Altogether, these results highlight several potentially non-obvious insights about conducting the type of tests that we propose.
While the algorithms appeared fairly similar in the vaccine case study, this mortgage allocation setting is 
more challenging in various ways: the presence of reporting behavior where some reports do not actually correspond to ``true'' harm; 115 total groups compared to 29; and much smaller numerical gaps, i.e. smaller group sizes, both for base preponderances $\Basegroup$ and reporting rates $\mu_G$. 
These additional challenges reveal some practical takeaways for conducting these tests. 

The first consideration is in the choice of algorithm. It appears that the betting-style test most effectively balances stopping time with identifying highly-risky groups---though it tends to stop more slowly than the asymptotically-valid Z-test, it also identifies groups that are more severely harmed (while also preserving true statistical validity). On the other hand, though the finite-sample Z-test appears to have similar theoretical guarantees as the betting-style test, it stops more slowly and in general appears to be much more likely to fail when gaps are smaller. 
This leads to the second consideration, which is the choice of $\beta$. Because it is statistically valid to retest with increasing values of $\beta$, these results suggest that the initial $\beta$ should be fairly small, and increased over time---especially as these tests tend to be fairly conservative. 
}

%% file: 6-discussion.tex
\section{Discussion}

This work is an initial approach to using reporting databases for post-deployment auditing; we believe there is a rich range of future work that develops the ideas in this paper, both technically and conceptually. 
\iftoggle{icml}{
On the statistical and algorithmic side, because our framework allows for plugging in any existing sequential test, new methods that control for multiple hypothesis testing both over time and over the number of distinct hypotheses would be directly beneficial for this application. More conceptually, for reporting databases to be practically useful, there are a plethora of additional considerations to incorporate from a variety of disciplines. For instance, if a reporting system was available, how would individuals engage with them in theory, and in practice?
}{

On the statistical and algorithmic side, because our framework allows for plugging in any existing sequential test, new methods that control for multiple hypothesis testing both over time and over the number of distinct hypotheses would be directly beneficial for this application. On the other hand, one might hope for online methods that do not require pre-specifying hypotheses and instead develops them sequentially in a quasi-unsupervised fashion, or that improve guarantees by exploiting relationships across hypotheses, as has proven useful in multi-objective learning. 

More conceptually, 
while the application examples in Section \ref{sec:experiments} are somewhat stylized, they demonstrate that reporting databases can be promising starting points for new types of post-deployment evaluation. 
For reporting databases to be practically useful, there are a plethora of additional considerations to incorporate from a variety of disciplines. For instance, if a reporting system was available, how would individuals engage with them in theory, and in practice? To what extent do, and should, individual incentives affect the database, and how it is designed? How can the result of a test (a null hypothesis rejection) be contextualized by existing and emerging legal frameworks? 
}

To the best of our knowledge, we are the first to propose individual incident reporting to identify patterns of disproportionate harm in interactions with a particular system; more generally, however, one might imagine that similar reporting systems can be developed to provide insights about concerns beyond fairness. 
In fact, while the framework introduced in our work is not intrinsically about algorithmic deployments, it is one way to operationalize recent regulatory movement in AI policy towards allowing for or requiring individual reports. Any way to make such reports actionable at large scale must, to some extent, aggregate of individual reports to develop more systematic evaluations of an underlying algorithm. We therefore see our work as one step towards giving voice to individual experiences---and towards having them make a difference. 

%% file: a1-proofs.tex
\section{Omitted Proofs}
\label{app:seq-proofs}
\iftoggle{icml}{
\subsection{Omitted proofs from Section \ref{sec:model}}
We prove Proposition \ref{prop:relativerisk-conversion}, restated below. 
\propRR*
\begin{proof}[Proof of Proposition \ref{prop:relativerisk-conversion}]
First, note that by definition of $\rho$, $\rho_G$, and $\RR_G$, we have 
\[
\rho_G \leq b \cdot \rho \iff \frac{\Pr[R = 1 \mid G ]}{\Pr[Y = 1 \mid G]} \leq b \cdot \frac{\Pr[R = 1]}{\Pr[Y = 1]} \iff \RR_G \geq \frac{\Pr[R = 1 \mid G]}{\Pr[R = 1]} \cdot \frac{1}{b}. 
\]
By Bayes' rule, $\frac{\Pr[R = 1 \mid G]}{\Pr[R = 1]} = \frac{\Pr[ G \mid R = 1]}{\Pr[G]} = \frac{\mu_G}{\Basegroup}$; furthermore, by assumption, we have $\frac{\mu_G}{\Basegroup} \geq \beta$. 
The result follows from combining with the previous display. 
\end{proof}
We prove Proposition \ref{prop:reporting-conversion}, restated below. 
\propIR*
\begin{proof}[Proof of Proposition \ref{prop:reporting-conversion}]
Recall that we have defined $\mu_G = \Pr[G \mid R]$, and $\Basegroup = \Pr[G]$ is known by Assumption \ref{assn:ref}.
By Bayes' rule, we have
$    \mu_G = \Pr[G \mid R] = \frac{\Pr[G]\Pr[R \mid G]}{\Pr[R]} =  \Basegroup\frac{\Pr[R \mid G]}{r},$
Now, let us decompose $\Pr[R \mid G]$ by ``true'' reports ($\Badevent = 1$) and ``false'' reports ($\Badevent = 0$). 
By the law of total probability,
$
    \Pr[R \mid G] 
    = r \cdot \left(\Truereportrate \Actualgroup + \Falsereportrate(1-\Actualgroup )\right)
$; more precisely, 
\begin{align*}
    \frac1r\Pr[R \mid G] &= \Pr[R \mid G, \Badevent = 1]\Pr[\Badevent \mid G]  + \Pr[R \mid G,  \Badevent =0](1-\Pr[\Badevent \mid G] )
    \\&= \Truereportrate \Actualgroup + \Falsereportrate(1-\Actualgroup )
    \\&= \Falsereportrate + \Actualgroup (\Truereportrate-\Falsereportrate);
\end{align*} 
combining this with the Bayes' rule computation, cancelling the $\frac1r$ factor, gives us $
    \Actualgroup  = \frac{\frac{\mu_G}{\Basegroup} - \Falsereportrate}{\Truereportrate-\Falsereportrate}.
$
The result follows from the assumption that $\nicefrac{\mu_G}{\Basegroup} \geq \beta.$
\end{proof}
}{}

\subsection{Omitted proofs for Sequential Z-test}
We prove Theorem \ref{thm:validity_ztest}, restated below. 
\ztestvalidity*

To prove this result, we will use a foundational result known as Ville's inequality \citep{ville1939etude}.
\begin{theorem}[Ville's inequality]
\label{thm:ville}
    Let $\{M_t\}_{t \in \mathbb{N}^+}$ be a non-negative supermartingale, i.e. for all $t$, $M_t \geq 0$, and $\E[M_{t+1} \mid \mathcal{F}_t] \leq M_t$, where $\mathcal{F}_t$ is the filtration (history) of all realizations of randomness up to and including time $t$. Then, for any $x \in (0,1)$, we have $\Pr[\exists t: M_t > \nicefrac{\E[M_0]}{x}] \leq x$.
\end{theorem}

The central thrust of our proof of Theorem \ref{thm:validity_ztest} is due to \citet{koolen2017quick} (which itself draws from \citet{balsubramani2014sharp}, and is a refinement of \citet{jamieson2014lil}); we reproduce the argument in the context of our work below, though we emphasize that we do not claim the proof technique as ours.
\begin{proof}[Proof of Theorem \ref{thm:validity_ztest}]
It is sufficient to show that for any group $G$ where $\NullH$ holds, we have $\Pr[ \exists t: G \in\FlagG] \leq \nicefrac{\alpha}{|\Groups|}$; the statement of the theorem follows from the Bonferroni correction over all $|\Groups|$ hypotheses. 

Ville's inequality (Theorem \ref{thm:ville}) appears similar in form to the statement we hope to prove; we therefore seek to transform our test statistic $\Logwealth_t^G = \sum_{s \in [t]}\1[X_s \in G]$ into a quantity that can be interpreted as a (non-negative) supermartingale. Although $\{\Logwealth_t^G\}_{t \in \mathbb{N}^+}$ is by itself clearly not a non-negative supermartingale, each $\Logwealth_t^G$ is the sum of $t$ Bernoulli trials with mean $\mu_G$, and Bernoulli random variables are sub-Gaussian with variance parameter $\nicefrac14$. Each $\Logwealth_t^G$ therefore satisfies the property that $\E[\exp(\eta(\Logwealth_t^G - \E[\Logwealth_t^G])] \leq \exp(\eta^2/8)$. 

This holds for any $\eta$, so we will construct a distribution $\phi$ on $\eta$ and use it to construct a martingale $M_t$. In particular, note that under $\NullH$, $\E[\Logwealth_t^G] < t \cdot \beta\Basegroup$. Thus, we 
let $S_t \coloneqq \Logwealth_t^G - \E[\Logwealth_t^G] =  \Logwealth_t^G - t \beta\Basegroup$. 
We will let 
$M_t = \int \phi(\eta) \exp(\eta S_t - t\eta^2/8) d\eta$.
Then, for any distribution $\phi$, $\{M_t\}_{t \in \mathbb{N}^+}$ is a non-negative supermartingale with respect to the randomness in realizations of reports $X_t$.
To see this, we have 
\begin{align*}
    \E[M_{t+1} \mid \mathcal{F}_t] &= \E\left[\int \phi(\eta)\exp\left(\eta(S_t + \1[X_{t+1} \in G] - \beta\Basegroup) - \tfrac{(t+1)\eta^2}{8}\right) d\eta \bigmid \mathcal{F}_t\right]
    \\&= \int \phi(\eta) \exp\left(\eta S_t - \tfrac{t\eta^2}{8}\right) \E\left[\exp\left(\eta(\1[X_{t+1} \in G] - \beta\Basegroup) - \tfrac{(t+1)\eta^2}{8}\right) \bigmid \mathcal{F}_t\right] d\eta 
    \\&\leq \int \phi(\eta) \exp\left(\eta S_t - \tfrac{t\eta^2}{8}\right) d\eta 
    \\&= M_t,
\end{align*}
where the inequality is due to $\frac1t\E[\Logwealth_t^G] \leq \beta\Basegroup$ and subgaussianity. It thus remains to use this martingale to compute an appropriate threshold $\theta_t(\alpha)$ on $\Logwealth_t^G$. 

$M_t$ will satisfy the conditions of Theorem \ref{thm:ville} for any choice of $\phi$, including one which puts point mass of 1 on $\eta = \eta'$ and 0 elsewhere, i.e. $\phi(\eta') = 1$ and $\phi(\eta) = 0$ for any $\eta \neq \eta'$. One path towards establishing the threshold $\theta_t(\alpha)$ is to simply pick one value of $\eta$; however, such an $\eta$ cannot depend on $t$ and would thus result in a suboptimal threshold. 
Instead, we will construct $\phi$ such that it is a discrete distribution, indexed by $i \in \mathbb{N}^+$, so that $\eta$ takes values $\eta_1, \dots, \eta_i$ with probability $\phi_1, \dots, \phi_i$; this allows each $\eta_i$ to depend on $t$ and therefore more finer-grained optimization of the threshold. Before committing to the exact distribution $\phi$, we first illustrate how $\phi_i$ and $\eta_i$ will be used in the threshold. 

Note that $M_t = \sum_{i \in \mathbb{N}^+} \phi_i\exp(\eta_iS_t - t\eta_i^2/8) \geq \max_i  \phi_i\exp(\eta_iS_t - t\eta_i^2/8)$, so for any $\delta$, we have 
\[
\{M_t \geq 1/\delta\} \supseteq  \{\max_i  \phi_i\exp(\eta_iS_t - t\eta_i^2/8) > 1/\delta \} = \left\{S_t \geq \min_i \left(\frac{t \eta_i}{8} + \frac{1}{\eta_i}\ln \frac{1}{\phi_i\delta}\right)\right\}, 
\]
and thus, picking $\theta_t(\alpha) = t \beta\Basegroup + \min_i \left(\frac{t \eta_i}{8} + \frac{1}{\eta_i}\ln \frac{1}{\phi_i\nicebonf}\right)$ would guarantee that $\Pr[\exists t: \Logwealth_t^G >  \theta_t(\alpha)] \leq \nicebonf$. 

Finally, we must commit to $\phi_i$, $\eta_i$, then optimize for $i$. Let $\phi_i = \frac{1}{i(i+1)}$ (note that $\sum_i \phi_i = 1$, so this is a valid distribution), $\eta_i = 2\sqrt{\frac{2\ln(1/\phi_i(\nicebonf))}{2^i } }$, and $i = \lfloor \log_2(t) \rceil$. For $i = \log_2(t)$ (without rounding), this would have yielded $\eta_i = 2 \sqrt{\frac{2 \ln((\log_2(t) + 1)(\log_2(t))/(\nicebonf))}{t}}$ and $\theta_t(\alpha) = \frac12\sqrt{2t \ln((\log_2t)(\log_2t + 1)/\nicebonf)}$. 
The statement follows from handling the numerical impact of rounding. 
\end{proof}
\begin{remark}
    A key constant in the proof of the version of the algorithm that is valid in finite samples is the subgaussian variance parameter, for which we used $\nicefrac14$ (and which propagates to a multiplicative factor of $\sqrt{1/4} = 1/2$ on the threshold). This is because the variance \textit{any} Bernoulli is at most $\nicefrac14$; however, this also motivates the choice of constant for the asymptotically-valid version of the test, which instead uses the variance parameter $\beta\Basegroup(1-\beta\Basegroup)$. 
\end{remark}

We now prove the power result. 
\ztestpower*
\begin{proof}
Let $\Gstar \coloneqq \arg\max_{G \in\Groups} \mu_G - \beta\Basegroup$ and let $\Delta \coloneqq \mu_{\Gstar} - \beta\Basestar$. Without loss of generality, we can consider only the test corresponding to $\Gstar$ (while still testing at level $\nicefrac{\alpha}{|\Groups|}$). 
Recall that for this instantiation of Algorithm \ref{alg:abstract}, the test statistic $\Logwealth_t^\Gstar = \sum_{s \in [t]} \1[X_s \in \Gstar]$ is simply the number of all reports belonging to $\Gstar$ by time $t$, and that stopping time $T$ is the first time where $\Logwealth_t^\Gstar$ surpasses the threshold $\theta_t(\alpha)$, i.e., $T \coloneqq \inf_{t \in {\mathbb{N}}^+} \Logwealth_t^\Gstar > t \beta\Basestar + \tfrac12 \sqrt{2.06 t \ln \left(|\Groups| \frac{(2 + \log_2(t))^2}{\alpha}\right)}$. For ease of notation, we will denote $C_1 \coloneqq \tfrac12\sqrt{2.06} = 0.718$ within this proof. 

For the first claim, it is sufficient to show $\liminf_{t \to \infty} \Pr[T > t] = 0$.\footnote{For a simple proof of this fact, see the solution to Problem 1.13 in \citet{bertsekas2008introduction}.} 
Recall that, by our modeling, we can consider $\Logwealth_t^\Gstar$ to be the sum of $t$ i.i.d. Bernoulli trials with parameter $\mu_\Gstar$.
Applying Hoeffding's inequality to this sum yields for any $t$ that
\begin{align*}
\Pr[T > t] &= 
\Pr\left[\Logwealth_t^\Gstar < t\beta\Basestar + C_1 \sqrt{t \cdot \ln \left(|\Groups| \frac{(2 + \log_2(t))^2}{\alpha}\right)}\right]
% \\&= \Pr\left[\E[\Logwealth_t^\Gstar] - \Logwealth_t^\Gstar > \Delta t + C_1 \sqrt{t \cdot \ln \left(|\Groups| \frac{(2 + \log_2(t))^2}{\alpha}\right)} \right]
% \\&\leq \exp\left(-\frac{2(t^2\Delta^2 + t C_1^2\ln(\frac{(2 + \log_2(t))^2}{\alpha}) - 2t\Delta C_1\sqrt{t \cdot \ln\left(|\Groups| \frac{(2 + \log_2(t))^2}{\alpha}\right)})}{t} \right)
\\&\leq \exp\left(-2\left(\Delta^2t - 2\Delta C_1 \sqrt{t \cdot \ln \left(|\Groups| \frac{(2 + \log_2(t))^2}{\alpha}\right)}\right) \right).
\end{align*}
Note that $\frac{\sqrt{t}\ln(\log_2(t))}{t} \to 0$; it can thus be seen that $\lim_{t \to\infty} \Pr[T > t] = \lim_{t \to \infty} \exp(-t) = 1$.

For the second claim, we apply Hoeffding's inequality again to see that for all $t$,  
\[
\Pr\left[\Logwealth_t^\Gstar \leq \E[\Logwealth_t^\Gstar] - C_1  \sqrt{ t \ln\left(\frac{(2 + \log_2(t))^2}{\delta}\right)}\right] \leq \Pr\left[\Logwealth_t^\Gstar \leq \E[\Logwealth_t^\Gstar] - \sqrt{\frac t2 \ln\left(\frac{ 1}{\delta}\right)}\right] \leq \delta.
\]
Thus, with probability at least $1-\delta$, for all $t$ simultaneously, $\Logwealth_t^\Gstar > t\mu_\Gstar - C_1  \sqrt{ t \ln\left(\frac{(2 + \log_2(t))^2}{\delta}\right)}$. 
The algorithm stops at time $t$ if and only if 
\[
t \mu_\Gstar - C_1  \sqrt{ t \ln\left(\frac{(2 + \log_2(t))^2}{\delta}\right)} > t \beta\Basestar + C_1  \sqrt{ t \ln\left(\frac{(2 + \log_2(t))^2}{\nicebonf}\right)}. 
\]
Rearranging, we have 
\[
\frac{t}{\left(\sqrt{\ln\left(\frac{(2 + \log_2(t))^2}{\nicebonf}\right)} + \sqrt{\ln\left(\frac{(2 + \log_2(t))^2}{\delta}\right)}\right)^2} \geq \frac{C_1}{\Delta^2}.
\]
Note that we can can upper bound the denominator of the left hand side as 
% \begin{align*}
$    \left(\sqrt{\ln\left(\frac{(2 + \log_2(t))^2}{\nicebonf}\right)} + \sqrt{\ln\left(\frac{(2 + \log_2(t))^2}{\delta}\right)}\right)^2 \leq 4 \ln\left(\frac{(2 + \log_2(t))^2}{\min(\nicebonf, \delta)}\right). $
% \end{align*}
Setting $\frac{t}{4 \ln\left(\frac{(2 + \log_2(t))^2}{\min(\nicebonf, \delta)}\right)} \geq \frac{C_1}{\Delta^2}$ and rearranging gives 
\begin{align}
\label{eq:tildeT}
\frac{t}{1 + \ln((2 + \log_2(t))^2)} \geq \frac{4C_1\ln(\max(\nicebonf, 1/\delta)}{\Delta^2}  
\end{align}
Thus, with probability $1-\delta$, the algorithm terminates at the smallest $t$ satisfying Equation \eqref{eq:tildeT}. The statement of the theorem follows from separating the two cases for $\delta < \nicebonf$ and $\delta \geq \nicebonf$, and noting that $\widetilde{O}$ notation suppresses the (negligible) log-log factor. 
\end{proof}

\subsection{Omitted proofs for betting-style algorithm}
\label{app:eval}

We first prove Theorem~\ref{thm:validity_evals}, restated for the sake of presentation.
\evalsvalidity*

\begin{proof} First note that for any $G$ for which $\NullH$ holds, the sequence $\{\exp(\Logwealth_t^G)\}_{t\geq 0}$ is a non-negative super-martingale. The non-negative property follows directly from the quantity being an exponential of a real (albeit possibly negative) number, while the fact that it is a super-martingale follows from the computations below:
\begin{align*}
    \E[\exp(\Logwealth_t^G)|\mathcal{F}_t] 
    &= \E[\exp(\Logwealth_{t-1}^G + \ln(1+\lambda_t^G(\1_{X_t\in G} - \beta \Basegroup)))|\mathcal{F}_t]
    \\&= \exp(\Logwealth_{t-1}^G) \cdot (1+\lambda_t^G(\E[\1_{X_t\in G}|\mathcal{F}_t] - \beta \Basegroup)) 
    \\&= \exp(\Logwealth_{t-1}^G) \cdot (1+\lambda_t^G(\mu_G - \beta \Basegroup) )
    \\&\leq \exp(\Logwealth_{t-1}^G) \cdot (1+\lambda_t^G(\beta \Basegroup - \beta \Basegroup) )
    \\&= \exp(\Logwealth_{t-1}^G),
\end{align*}
where the first equality follows by Eq.~\ref{eq:wealth_update}, the second by re-arranging and noting that all quantities except $\mathbf{1}_{X_t\in G}$ are completely determined by $\mathcal{F}_t$\footnote{In particular, it is imposed that $\lambda_t^G$ be 'predictable' which precisely implies that it is fixed given $\mathcal{F}_t$.}, and the third by definition (see Section~\ref{sec:model}). Finally, the inequality follows because $\mu_G \leq \beta\Basegroup$ under $\NullH$ and $\lambda_t^G \geq 0$.

Next, for any group $G$ such that $\NullH$ holds, we can apply Ville's inequality (Theorem~\ref{thm:ville}), plugging in the super-martingale $\{\exp(\Logwealth_t^G)\}_{t\geq 0}$ and taking $x$ to be $\theta_t(\alpha) = \log{(|\Groups|/\alpha)}$. This yields the following guarantee: 
\begin{align*}
\Pr[\exists t: \Logwealth_t^{G} > \log(|\Groups|/\alpha)] &= \Pr[\exists t: \exp(\Logwealth_t^{G}) > |\Groups|/\alpha] 
\\&\leq \E[\exp(\Logwealth_0^{G})] \cdot \alpha/|\Groups| 
\\&= \alpha/|\Groups|,
\end{align*}
where the final line follows because $\Logwealth_0^{G}$ is initialized as $0$ and hence $\exp(\Logwealth_0^{G})$ is equal to $1$.

Finally, by union bound we get the desired guarantee:
\begin{align*}
\Pr[\exists t: \exists G \in \FlagG \text{ s.t. } \NullH \text{ holds}] &\leq \sum_{G \text{ s.t. }\NullH \text{ holds}} \Pr[\exists t: \Logwealth_t^G > \log{(|\Groups|/\alpha)}] \\
&\leq |G \text{ s.t. }\NullH \text{ holds}| \cdot \alpha / |\Groups|\\
&\leq \alpha.
\end{align*}
\end{proof}

Before proving Theorem~\ref{thm:power_evals}, we first state and prove some helper results. 

\begin{claim}
\label{claim:logwealth-regret} For any $T\geq 4$ and group $G$, we have that the expected value over the randomness in the realizations of each $X_t$ of $\Logwealth_T^G$  defined as per Equations~\eqref{eq:wealth_update} and \eqref{eq:bet_update} can be lower bounded as
\[
\E[\Logwealth_T^G] \geq \E\left[\max_{\lambda\in[0, 1]} \Logwealth_T(\lambda)\right] - 2 \ln T.
\]
where we define $\Logwealth_T^G(\lambda)$ to be the quantity obtained by applying Equation~\eqref{eq:wealth_update} with $\lambda_t^G \coloneqq \lambda$ for all $t\in[T]$.
\end{claim}
\begin{proof} By the definition of regret we have that 
$\max_{\lambda\in[0, 1]} \Logwealth_T^G(\lambda) - \Logwealth_T^G \leq R_T$.
Rearranging and taking expectations, we have 
\[ \E[\Logwealth_T^G] \geq \E\left[\max_{\lambda\in[0, 1]} \Logwealth_T^G(\lambda)\right] -  \E[R_T].\]
Next, it can be verified that Equation~\eqref{eq:bet_update} implements the Online Newton Step algorithm for $\ln(1 + \lambda_t^G(\1_{X_t \in G} - \beta\Basegroup))$ (see Appendix C of \citet{cutkosky2018black}). We therefore have that $R_T \leq \frac{1}{2-\ln(3)}\ln(T  + 1)$ in general, and $R_T \leq 2\ln(T)$ for $T \geq 4.$ 
The statement of the claim plugging this into the expression above.
\end{proof}

\begin{lemma}\label{lem:lambda_opt} 
For each group $G$, taking $\lambda^\star_G = \mathrm{Proj}_{[0,1]}\left[\dfrac{\mu_G - \beta \mu_G^0}{\beta\mu_G^0(1-\beta\mu_G^0)}\right]$ maximizes expected log-wealth (at every step $t$). The resulting expected log-wealth at time $T$ (had $\lambda_G^\star$ been used at every time $t$) is equal to
\[
\E\left[\Logwealth_T^G(\lambda^\star_G)\right] = T\cdot \Logwealth_\star^G 
\]
where we denote $\Logwealth_\star^G \coloneqq \E[\ln(1+\lambda_G^\star(\mathbf{1}_{X_t\in G} - \beta\Basegroup)]$ the expected one-step wealth change under the bet $\lambda^\star_G$.
\end{lemma}

\begin{proof}
For a fixed $\lambda$, the log-wealth at time $T$ is given by 
\begin{equation*}
\Logwealth_T^G(\lambda) = N_T \ln{(1+\lambda(1-\beta\mu_G^0))} + (T-N_T) \ln{(1-\lambda\beta\mu_G^0)},
\end{equation*}
where $N_T = \sum_{t=1}^T \mathbf{1}_{X_t\in G}$. Taking expectations, we have that $\E[N_T] = T\cdot \mu_G$ and therefore 
\begin{equation}\label{eq:expected_logwealth}
\E\left[\Logwealth_T^G(\lambda)\right] = T\cdot\left[\mu_G \ln{(1+\lambda(1-\beta\mu_G^0))} + (1-\mu_G) \ln{(1-\lambda\beta\mu_G^0)}\right]. 
\end{equation}
To maximize \eqref{eq:expected_logwealth}, we only need to find $\lambda_G^\star\in[0,1]$ that maximizes the expressions in the square brackets. Taking the derivative we see that the function is concave,
and, therefore, we can solve for $\lambda_G^\star$ by setting the derivative to $0$ and then projecting the resulting value to $[0,1]$. This yields
$\lambda^\star_G = \text{Proj}_{[0,1]}\left[\frac{\mu_G - \beta \mu_G^0}{\beta\mu_G^0(1-\beta\mu_G^0)}\right].$ Plugging this back into \eqref{eq:expected_logwealth} we get
\begin{align*}
\E\left[\Logwealth_T^G(\lambda_G^\star)\right] &= T\cdot \left[\mu_G \ln{(1+\lambda_G^\star(1-\beta\mu_G^0))} + (1-\mu_G) \ln{(1-\lambda_G^\star\beta\mu_G^0)}\right] \\
&= T\cdot \E[\ln(1+\lambda_G^\star(\mathbf{1}_{X_t\in G} - \beta\Basegroup)] \\
&\coloneqq T\cdot \Logwealth_\star^G.
\end{align*}
\end{proof}

\begin{remark}
    Note that we can explicitly compute 
\[\Logwealth_\star^G = \mu_G \ln\left(    1 + \tfrac{\Delta_G}{\beta\Basegroup(1 - \mu_G )}\right) + \ln\left(1 - \tfrac{\Delta_G}{1-\beta\Basegroup}\right),\]
where $\Delta_G = \mu_G - \beta\Basegroup$, but this quantity is difficult to analyze, and it is not clear that $\Logwealth_\star^G$ can be explicitly lower bounded as $O(\Delta_G)$. 
\end{remark}

We now prove Theorem~\ref{thm:power_evals}, which we restate below.

\evalspower*

\begin{proof}
Let $\Gstar \coloneqq \arg\max_G \Logwealth_\star^G$ and denote the corresponding one-step wealth change $\Logwealth_\star = \Logwealth_\star^\Gstar$. Note that under the alternative this will correspond to a strictly positive quantity and is equivalent to the definition in the theorem statement. We can analyze the likelihood that its null has not been rejected by time $t$ as follows:
\begin{align*}
\Pr\left[\Logwealth_t^{\Gstar} < \ln(\nicebonfinv)\right] 
&=  \Pr\left[\Logwealth_t^{\Gstar} - \E[\Logwealth_t^{\Gstar}]< \ln(\nicebonfinv) - \E[\Logwealth_t^{\Gstar}]\right] \\ 
 &\leq \Pr\left[\Logwealth_t^{\Gstar} - \E[\Logwealth_t^{\Gstar}]< \ln(\nicebonfinv) - (t\cdot \Logwealth_\star - 2\ln t)\right], 
 \end{align*}
 where the inequality follows by Claim \ref{claim:logwealth-regret} and Lemma~\ref{lem:lambda_opt}, and the fact that $\E[\max_{\lambda\in[0, 1]} \Logwealth_t^\Gstar(\lambda)] \geq \E[\Logwealth_t^\Gstar(\lambda_\star^\Gstar)] = t \cdot \Logwealth_\star$. Whenever $t$ is large enough such that $\frac{\ln(t)}{t} \leq \frac{\Logwealth_\star}{4}$, we have 
\begin{align}
\label{eq:prnostop}
    \Pr[\Logwealth_t^\Gstar < \ln(\nicebonfinv)] \leq \Pr\left[\Logwealth_t^{\Gstar} - \E[\Logwealth_t^{\Gstar}]< \ln(\nicebonfinv) - \tfrac{3}{4}(t\cdot \Logwealth_\star)\right].
\end{align}
Since $\sqrt{t} \geq \ln{t}$ for all $t\in \mathbb{N}^\star$, this is satisfied in particular by taking $t\geq \frac{2^4}{\Logwealth_\star^2}$. Further, note that $\ln(\nicebonfinv) \leq \frac{t\cdot \Logwealth_\star}{4}$ whenever $t\geq \frac{2^2\cdot \ln{(\nicebonfinv)}}{\Logwealth_\star}$. So, for $t\geq \max\{\frac{2^4}{\Logwealth_\star^2}, \frac{2^2\cdot \ln{(\nicebonfinv)}}{\Logwealth_\star}\}$, we have
\[\Pr[\Logwealth_t^\Gstar < \ln(\nicebonfinv)] \leq \Pr\left[\Logwealth_t^{\Gstar} - \E[\Logwealth_t^{\Gstar}]< - \tfrac{1}{2}(t\cdot \Logwealth_\star)\right].\]
Now, note that since $\lambda_t^G \in [0,1]$, we have that each $\ln(1 + \lambda_t^G(\1_{X_t \in G} - \beta\Basegroup))$ lies in $[\ln{(1-\beta\Basegroup)}, \ln{(2 - \beta\Basegroup)}]$ and is therefore sub-Gaussian with parameter $\sigma = \frac{1}{2}\ln{\left(1+ \frac{1}{1-\beta\Basegroup}\right)}$; then, Hoeffding's inequality gives 
\begin{align*}
\Pr\left[\Logwealth_t^{\Gstar} - \E[\Logwealth_t^{\Gstar}]< - \tfrac{1}{2}(t\cdot \Logwealth_\star)\right]
&= 
\Pr\left[\sum_{i \in [t]} \ln(1 + \lambda_t^\Gstar(\1_{X_t \in \Gstar} - \beta\mu_\Gstar^0)) - \E[\Logwealth_t^{\Gstar}] \leq -\frac{1}{2} t \cdot \Logwealth_\star \right] \\
&\leq \exp\left(-\frac{(\frac{1}{2} t \cdot \Logwealth_\star)^2}{\frac{1}{2}t \ln^2(1+ \frac{1}{1-\beta\mu_\Gstar^0})}\right)  
\\&= \exp\left(-\frac{\Logwealth_\star^2}{2\ln^2(1+ \frac{1}{1-\beta\mu_\Gstar^0})} \cdot t \right) \\
&\leq  \exp\left(- \frac{(1-\beta\mu_\Gstar^0)^2}{2} \cdot \Logwealth_\star^2 \cdot t \right). 
\end{align*}
where for the last inequality we used $\ln{(1+x)} \leq x$. Now we are ready to analyze the stopping time $T$ of Algorithm~\ref{alg:abstract}. 

\paragraph{Test of power one.} Let $E_t$ be the event that we stop at time $t$, i.e. $E_t = \{\exists G$ such that $\Logwealth_t^G \geq \nicebonfinv\}$. We have that
\begin{align*}
\Pr[T = \infty] &= \Pr\left[\lim_{t \to \infty} \cap_{s \leq t} \neg E_t\right] \\
&= \lim_{t \to \infty} \Pr[ \cap_{s \leq t} \neg E_t] \\
&\leq \lim_{t \to \infty} \Pr[ \neg E_t] \\
&= \lim_{t \to \infty} \Pr[\forall G, \, \Logwealth_t^G < \ln(\nicebonfinv)] \\
&\leq \lim_{t \to \infty}\Pr\left[\Logwealth_t^{\Gstar} < \ln(\nicebonfinv)\right] \\
&\leq \lim_{t \to \infty}\exp\left(- \frac{(1-\beta\Basegroup)^2}{2} \cdot \Logwealth_\star^2 \cdot t \right) \\
&= 0.
\end{align*}

\paragraph{Expected Stopping Time.} Since $T$ is a positive integer, we can express the expected stopping time as
\begin{align}
\E[T] &= \sum_{t=1}^\infty \Pr[T>t] \notag \\
&= \sum_{t=1}^\infty \Pr[\neg E_1 \land \ldots \land \neg E_t] \notag\\
&\leq \sum_{t=1}^\infty \Pr[\neg E_t] \notag\\
&= \sum_{t=1}^\infty \Pr[\forall G, \, \Logwealth_t^G < \ln(\nicebonfinv)]\notag \\
&\leq \sum_{t=1}^\infty \Pr\left[\Logwealth_t^{\Gstar} < \ln(\nicebonfinv)\right] \notag\\
&\leq \max\bigg\{\frac{2^4}{\Logwealth_\star^2}, \frac{2^2 \cdot \ln(\nicebonfinv)}{\Logwealth_\star}\bigg\} + \sum_{t=1}^\infty \exp\left(- \frac{(1-\beta\mu_\Gstar^0)^2}{2} \cdot \Logwealth_\star^2 \cdot t \right) \label{eq:our_up}\\
&= \max\bigg\{\frac{2^4}{\Logwealth_\star^2}, \frac{2^2 \cdot \ln(\nicebonfinv)}{\Logwealth_\star}\bigg\} + \frac{1}{\exp{((1-\beta\mu_\Gstar^0)^2 \Logwealth_\star^2 /2)} - 1} \notag\\
&\leq \max\bigg\{\frac{2^4}{\Logwealth_\star^2}, \frac{2^2 \cdot \ln(\nicebonfinv)}{\Logwealth_\star}\bigg\} + \frac{2}{(1-\beta\mu_\Gstar^0)^2 \Logwealth_\star^2} \label{eq:exp_ineq}\\
&\leq \mathcal{O}\left(\frac{1}{\Logwealth_\star^2} + \frac{\ln(\nicebonfinv)}{\Logwealth_\star}\right) \notag
\end{align}
where \eqref{eq:our_up} follows from the upper bound on $\Pr\left[\Logwealth_t^{\Gstar} < \ln(\nicebonfinv)\right]$ for $t \geq \max\left\{\frac{2^4}{\Logwealth_\star^2}, \frac{2^2 \cdot \ln(\nicebonfinv)}{\Logwealth_\star}\right\}$ derived in \eqref{eq:prnostop}, and \eqref{eq:exp_ineq} follows from $\exp(x) \geq 1+x$. 
\end{proof}

%% file: a0-practical.tex
\paragraph{Choosing $\Groups$.}
In our experiments in Section \ref{sec:experiments}, we choose to define subgroups as all possible combinations of available demographic characteristics.
That said, a practitioner may seek to define $\Groups$ more carefully in accordance with their application. For instance, if the goal is to illustrate discrimination in a legal sense, $\Groups$ should be defined with respect to (protected) demographic features, rather than arbitrary combinations of covariates. On the other hand, groups need not be solely demographic, which allows our approach to test for safety rather than solely fairness. For example, $\Groups$ could include which batch of a medication an individual received; our tests could then help identify whether some batches were improperly manufactured.  

\paragraph{Baseline rates $\{\Basegroup\}_{G \in \Groups}$.} A natural question that arises from the modeling in this section is how $\{\Basegroup\}_{G \in \Groups}$ can be determined, or if Assumption \ref{assn:ref} is strictly necessary. 
Practically speaking, these base preponderances may be estimated, possibly with some amount of noise; however, the estimation problem can be addressed with standard techniques and is not core to our contribution. 
Similarly, in practice these baseline preponderances may change over time (e.g. if some subgroups increased uptake of a vaccine, or applied for loans more frequently, over time); however, such situations are relatively straightforward to handle under our algorithmic frameworks (see, e.g., the variants discussed in \citet{chugg2024auditing}). 
We therefore focus on the case where we have access to the true, underlying values of $\{\Basegroup\}_{G \in \Groups}$ for ease and clarity of exposition. 

Note that testing against base preponderances of the reference population (i.e., to compare $\mu_G$ to $\Basegroup$) is a new test proposed by this work, and the analysis in Sections \ref{subsec:rr} and \ref{subsec:ir} is specific to this test. Existing approaches to monitoring in reporting databases compare to different baselines, most commonly the historical overall incidence rate for the specific symptom, sometimes by subgroup
\citep{shimabukuro2015safety, kulldorff2011maximized, oster2022myocarditis}. 
These baselines could, in principle, be plugged into the algorithms in Section \ref{sec:algs}, but new analysis for (possibly group-varying) reporting rates would be necessary to draw inferences about analogous quantities of interest (e.g., $\RR$ or $\mathrm{IR}$), as current approaches do not generally consider reporting behavior.
In contrast, our modeling allows us to identify what quantities may affect the true incidence rate even if they may be unmeasurable.

\paragraph{Setting $\beta$.} Finally, to run the test proposed in Equation \eqref{eq:htest}, it is necessary to determine how to set the value of $\beta$. As we will see in Section \ref{sec:algs}, when $\beta$ is set too high, then the test may never identify problematic groups, or identify them more slowly; on the other hand, as is clear from the previous subsections, rejecting the null hypothesis for a smaller $\beta$ requires more stringent assumptions on reporting behavior. Thus, we suggest a procedure to set $\beta$ as follows: (1) choose a relative risk or incidence rate threshold where it would be problematic for any group if $\RR_G$ or $\mathrm{IR}_G$ surpassed that threshold; (2) make the corresponding assumptions about reporting behavior; (3) use these quantities to compute a reasonable $\beta$.
We give some example computations in Section \ref{sec:experiments}. 
Due to an equivalence between hypothesis testing and confidence intervals, it is statistically valid to rerun tests with different $\beta$s once data collection has begun. Thus, it may be prudent to begin by setting the lowest $\beta$ that reporting assumptions would allow; then, if the tests appear to be stopping very quickly, to re-run them at higher $\beta$s, which would allow a practitioner to get a better sense of the severity of the harm.